\documentclass[10pt]{amsart}
\usepackage{amsfonts}

\usepackage{amsmath}
\usepackage{graphics}
\usepackage{epsfig}
\usepackage[colorlinks=true]{hyperref}
\usepackage{float}
\usepackage{caption}
\usepackage{subfigure}
\usepackage{wrapfig}

\newtheorem{theorem}{Theorem}[section]
\newtheorem{corollary}{Corollary}

\newtheorem{lemma}[theorem]{Lemma}

\newtheorem{conjecture}{Conjecture}

\theoremstyle{definition}

\newtheorem{remark}{Remark}

 \email{rparshad@clarkson.edu}

\begin{document}
\title[ prey cannibalism]{Prey cannibalism alters the dynamics of
Holling-Tanner type predator-prey models}
\author[  Basheer, Quansah, Bhowmick, Parshad ]{}
\maketitle

\begin{abstract}
Cannibalism, which is the act of killing and at least partial consumption of conspecifics, is ubiquitous in nature. Mathematical models
have considered cannibalism in the predator primarily, and show that predator cannibalism in two species ODE models provides a strong stabilizing effect.
There is strong ecological evidence that cannibalism exists among prey as well, yet this phenomenon has been much less investigated. In the current manuscript, we
investigate both the ODE and spatially explicit forms of a Holling-Tanner model, with ratio dependent functional response. We show that cannibalism in the predator provides a
stabilizing influence as expected.
However, when cannibalism in the prey is considered, we show that it cannot stabilise the unstable interior equilibrium in the ODE case, but can destabilise the stable interior equilibrium. In the spatially explicit case, we show that in certain parameter regime,
prey cannibalism can lead to pattern forming Turing dynamics, which is an impossibility without it. Lastly we consider a stochastic prey cannibalism rate, and
find that it can alter both spatial patterns, as well as limit cycle dynamics.
\end{abstract}

\medskip

\centerline{\scshape  Aladeen Basheer, Emmanuel Quansah, 
Schuman Bhowmick and Rana D. Parshad} 
\medskip

\centerline{Department of Mathematics,}  \centerline{Clarkson University,}  %
\centerline{ Potsdam, New York 13699, USA.}

\medskip
\section{Introduction}

\label{1} Cannibalism is the act of killing and at least partial consumption
 of conspecifics. This has been documented in human populations all around the world. It is/has been practiced in various indigenous South American and New Guinean tribes as a social norm \cite{online1}, and also reported in Africa \cite{pennell1991cannibalism}. It is also practiced among ``Aghoris" or witch doctors in Northern India, in the hope of gaining immortality \cite{online}. However, if one steps away from these seemingly macabre and ghoulish human practices, and observes real predator-prey systems, then cannibalism is actually commonplace \cite{ claessen2004population}. It is seen to actively occur in more than 1300
species in nature \cite{polis1981evolution}. 
Apart from being a fascinating socio-anthropological/ecological subject in its own right, cannibalism can lead to various ``counter intuitive" effects as concerns population dynamics. These include the \emph{life boat mechanism}, where the act of cannibalism causes persistence, in a population doomed to go extinct \cite{getto2005dis}. It can also induces stability, in an otherwise cycling population \cite{kohlmeier1995stabilizing}.

A detailed survey of the current and past mathematical literature on cannibalism, shows that various types of predator-prey models, with the inclusion of cannibalism, have been investigated. These include two species ODE models, three species ODE models incorporating stage structure, two species PDE models, discrete models and integro-differential equation models, as well as recent models of cannibalism that include diseased predators \cite{kohlmeier1995stabilizing, buonomo2010effect, biswas2014cannibalism, buonomo2006stabilizing, getto2005dis}.

The two species unstructured ODE models that have been investigated, are of the general type
\begin{align}
 \label{eq:1b} 
\frac{du}{dt}= u\left( 1-\frac{u}{K}\right) - f(u)v , 
\end{align}

\begin{align}
 \label{eq:2b}
\frac{dv}{dt}= - d v + \epsilon f(u)v  - cC(u,v),
\end{align}

where $u$ is the prey, and $v$ the predator. Here $cC(u,v)$ models cannibalism in the predator, with $c$ being the rate of cannibalism. 
In such models, cannibalism is seen to have a strong stabilising influence. 
Kohlmeir \cite{kohlmeier1995stabilizing} was the first to show this, by considering cannibalism in the predator, in
the classic Rosenzweig-McArthur model. That is when $f(u)$ is of Holling type II. They showed that predator cannibalism can surprisingly stabilise
the unstable interior equilibrium of the model.

We note that cannibalism, has been primarily modeled as an act in the \emph{predator species only} \cite{fasani2012remarks, sun2009predator, kohlmeier1995stabilizing}.
Note, there is strong evidence in ecology that cannibalism often occurs in the prey species \cite{rudolf2007interaction, rudolf2007consequences}. Recent experiments with prey cannibalism, show that they can have a significant impact on the population dynamics of predator-prey systems \cite{rudolf2008impact}. However, there is \emph{not a single model} in the mathematical literature, that considers the effect of prey cannibalism in two species unstructured ODE or unstructured PDE models. To the best of our knowledge there are very few works overall, that have considered prey canibalism, in some form or the other. Solis et al. considered the effect of prey cannibalism in a multi component structured model, via integer differential equations \cite{solis2014birth}. Chow and Jang investigated a structured discrete model
with prey cannibalism finding that it may either stabilize or destabilize, depending on cannibalism rates and various other
parameters \cite{chow2012cannibalism}. Also, Zhang et al. considered a three species structured ODE model, with structure and cannibalism in the prey. They showed that cannibalism strengthens the conditions of both global asymptotical stability, and persistence of the interior equilibrium point \cite{zhang2010rich}.

Many models of cannibalism with size and stage structure effects have been investigated. The approach in the literature has been to structure the predator class into adults and juveniles and allow predation of the juvenile predator by the adult predator. This essentially yields a three species ODE model. The reader is referred to \cite{buonomo2006stabilizing, buonomo2010effect, zhang2010rich}, and again therein cannibalism is seen to provide a stabilising influence, depending on parameter regime. However, depending on model structure and parameter regimes, cannibalism is also seen to destabilize \cite{claessen2004population} in some cases.
Various other classes of models such as integro-differential equation models, as well as discrete models have also been investigated \cite{chow2012cannibalism, solis2014birth, getto2005dis}.

Since species disperse in space in search of food, shelter, mates and to avoid predators, spatially dispersing populations are often modeled via partial
differential equations (PDE)/spatially explicit models of interacting species \cite{Gilligan1998, M93, okubo2001diffusion, sen2012bifurcation}. Of
particular interest to many is the phenomenon of Turing instability, first introduced by Alan Turing \cite{Turing52}, which shows that diffusion in a
system can lead to pattern forming instabilities. In the context of cannibalism, there are very few works in the literature that deal with
spatially explicit models and/or Turing instability. Sun \cite{sun2009predator} was the first to show that cannibalism in the
classical diffusive Rosenzweig-McArthur model, can bring about Turing patterns. This is an impossibility in the system without cannibalism. Fasani et al. \cite
{fasani2012remarks} considered the same model, but with a general form for the cannibalism term (of which the model of Sun et al. is a special case), 
and also found that cannibalism can bring about Turing patterns. The other direction does not seemed to have been studied. That is when and under what conditions, can cannibalism \emph{remove} Turing instability?
In the current manuscript we argue that these above approaches, do not capture the full breadth of cannibalism as it exists in nature. We cite the following reasons:

Many mathematical models of predator-prey systems, do not have a symmetric functional responses between
predator and prey. For example, the models of Leslie-Gower type \cite{leslie1960properties}, and Holling-Tanner type \cite{hsu1995global}. These models are ubiquitous in the mathematical literature (see \cite{banerjee2012turing, hsu1995global} and the references therein), however, the effect of cannibalism, in particular the effect of \emph{prey cannibalism} in such models, has not been considered.

Also, In all of the investigations in the literature, only a deterministic rate of cannibalism, $c$, has been considered. However, predator and prey species are often subject to uncertainties in the environment, modeled best via environmental noise. For these reasons it is possible that the various rates, (birth rate, death rate, predation rate etc.), one considers for modeling purposes, may not be known precisely, and may always involve uncertainty. Yet there are no studies where a model with a ``noisy" cannibalism rate has been considered, that is when $c \approx c + ``noise"$.
\\
\vspace{0.5mm}
\\
In the current manuscript we consider a Holling-Tanner predator-prey model, with a ratio dependent functional response \cite{A00}, that has been well studied in the literature \cite{banerjee2012turing}. Next we introduce into this model 
\newline
(a) cannibalism occurring only in the predator species,
\newline
(b) cannibalism occuring only in the prey species.
\newline
We investigate in detail:
\newline
(1) ODE versions modeling (a) and (b).
\newline
(2) The PDE, stochastic ODE  and stochastic PDE versions modeling (b).
\newline
\\
\vspace{0.5mm}
\\
Our contributions therein are the following:
\newline
1) We show that predator cannibalism, in certain parameter regime, can indeed stabilize the unstable interior equilibrium, in the ODE case, as expected from the known results on other models, with cannibalism in the predator.
\newline
2) We show that prey cannibalism \emph{cannot stabilize} the unstable interior equilibrium, in the ODE case. This is contrary to what is known about models with cannibalism in the predator.
\newline
3) We show that prey cannibalism can destabilize the spatially homogenous equilibrium in the PDE case. In particular, we show that in certain parameter regime such that there is no Turing instability in the Holling-Tanner model, introducing prey cannibalism \emph{can actually bring about Turing instability}. Alternatively, we show that in certain parameter regime, introducing predator cannibalism \emph{cannot bring about Turing instability}.
\newline
4) We show that in certain parameter regime, if Turing
instability does not exist in the model with prey cannibalism, then removing prey
cannibalism \emph{cannot induce} Turing instability.
\newline
5) The ODE version of the Holling-Tanner model we consider is known to possess a unique limit cycle.
We investigate the effects of a stochastic prey cannibalism rate $c$, on the limit cycle dynamics, finding that classical white noise cannot induce stability.
\newline
6) We also consider the effects of a stochastic prey cannibalism rate $
c$, on the spatial dynamics, finding that space-time white noise, can significantly alter the spatial dynamics of the model.


\section{The model formulation}
\subsection{The Holling-Tanner model}
We first present the following Holling-Tanner model, with ratio dependent functional response in prey,  \cite{banerjee2012turing}. 
\begin{equation}
 \label{eq:1}
\left. 
\begin{array}{c}
\frac{du}{dt}=u\left(1-u\right) -\dfrac{uv}{u+\alpha v} \\ 
\frac{dv}{dt}=\delta v\left( \beta -\dfrac{v}{u}\right)%
\end{array}%
\right\} nc-model 
\end{equation}
We will refer to the above model as the nc-model (no cannibalism) henceforth, and
recap certain dynamical aspects of the model. Note, \eqref{eq:1} is a nondimensionalised version of the model presented in \cite{banerjee2012turing}. The reader is referred to \cite{banerjee2012turing}, for a complete derivation on the parameters, their biological significance and the non-dimensionalization process.
Here $u$ and $v$ are the populations of the prey and predator species respectively.
 The dynamics of the prey species follows a ratio dependent functional response \cite{A00}.
The predator is modeled according to the Leslie-Gower scheme \cite{leslie1960properties}. This assumes that the carrying capacity of
the predator, is not constant but rather proportional to its food, the prey $u$. Hence a reasonable choice
for the carrying capacity $K=K(u)$ is $K(u)=u+d_{3}$. If $d_{3} >0$, the predator is thought of as a
generalist, that is it can change its food source in the absence of 
its favorite prey. If $d_{3}=0$, the predator is thought of as a specialist. In this case the resulting model is of Holling-Tanner type \cite{hsu1995global}.
Here $\beta$ is the intrinsic growth rate of the predator $v$, and $\alpha$ is essentially given by $\alpha= \frac{\mbox{prey \ handling \ time}} {\mbox{ prey \ capture \ rate }} $.

The non-trivial steady state $\left(u_{nc},v_{nc}\right) $ is given by:

\begin{equation*}
u_{nc}=1-\dfrac{\beta }{1+\alpha \beta },\text{ \ \ }v_{nc}=\beta \left( 1-%
\dfrac{\beta }{1+\alpha \beta }\right)  \label{5}
\end{equation*}

since we require

\begin{equation*}
 u_{nc}>0, \  \Longrightarrow \  1>\dfrac{\beta }{1+\alpha \beta} \ 
\end{equation*}
Thus we obtain the following restriction on $\beta$ due to a feasibility requirement,
\begin{equation*}
\beta <\dfrac{1}{1-\alpha }  \label{6}
\end{equation*}

%
%
%

For stability we require
 that the $Trace$ $J^{nc}<0$ which implies

\begin{equation*}
\dfrac{\beta }{\alpha \beta +1}+\dfrac{\beta }{\left( \alpha \beta
+1\right) ^{2}}-1<\beta \delta 
\end{equation*}

As the determinant $Det$ $J^{nc}>0$, is always positive, as we have $\Rightarrow \beta < \dfrac{1}{1-\alpha}$, which is always true as this is a feasibility condition.

%

\subsection{The Holling-Tanner model with predator cannibalism}

We now consider the inclusion of predator cannibalism in \eqref{eq:1}. This yields the following model

\begin{equation}
 \label{eq:prc}
\left. 
\begin{array}{c}
\dfrac{du}{dt}=u\left( 1-u\right) -\dfrac{uv}{u+\alpha v}  \\ 
\dfrac{dv}{dt}=\delta v\left( \beta_{1} -\dfrac{v}{\gamma u+cv}\right) 
\end{array}
\right\} predator \ cannibalism \ model 
\end{equation}

In deriving the above model we make the following assumptions 

\begin{itemize}
 \item The predator $v$ is depredating on the prey species $u$, as well as on its own species $v$.  
 
 \item Hence the food source of $v$ is $\gamma u+cv$, where $cv$ is the cannibalism term, with cannibalism rate $c$. Here we assume $c+\gamma \approx 1$, with $\gamma < 1, c < 1$. Thus the food intake of the predator stays roughly the same, even if it turns into a cannibal. It starts eating less prey $u$, so $\gamma < 1$, and now also eats some conspecifics, so $0 < c < 1$.

 \item There is a clear gain of energy to the cannibalistic predator, from the act of cannibalism. This gain results in an increase in reproduction in the predator. 
 This in turn leads to a gain in predator population, modeled via adding a $\beta_{1}v$ term to the predator equation. We therefore assume $\beta_{1} > \beta$. 
 \end{itemize}

The non-trivial steady state $\left( u_{c},v_{c}\right)$ is given by

\begin{equation}
u_{c }=\left( 1-\dfrac{\beta_{1} \gamma }{\alpha \beta_{1} \gamma -c\beta_{1} +1}\right),  \\ v_{c}=\left( \dfrac{\gamma \beta_{1} }{1-\beta_{1} c}\right) u_{c}
\end{equation}

since $u_{c}>0$ we obtain

$\left( 1-\dfrac{\beta_{1} \gamma }{\alpha \beta_{1} \gamma -c\beta_{1} +1}\right) >0$ 
and due to the assumptions and feasibility conditions, $\gamma < 1, \\ 1-\beta_{1} c > 0$ we have 

\begin{center}
\bigskip $\alpha \gamma -\gamma +\dfrac{1}{\beta_{1} }\ >c$
\end{center}

The Jacobian Matrix at $(u_{c},v_{c})$:

$J^{c}=\left( 
\begin{array}{ccc}
\allowbreak -u_{c}+\dfrac{\left( 1-\beta_{1} c+\alpha \gamma \beta_{1} \right)
^{2}}{\gamma \beta_{1} \left( 1-\beta_{1} c\right) } &  & \allowbreak -\dfrac{\left(
c\beta_{1} -1\right) ^{2}}{\left( \alpha \beta_{1} \gamma -c\beta_{1} +1\right) ^{2}} \\ 
&  &  \\ 
\allowbreak \beta_{1} ^{2}\gamma \delta &  & \allowbreak \beta_{1} \delta \left(
c\beta_{1} -1\right)%
\end{array}%
\right) 
$

 Notice that $J_{12}$ and $J_{22}$ are negative since $\gamma <1$ and $1-\beta_{1} c > 0$.

$\allowbreak $

 Det $J^{c}$=$\left( -u_{c}+\dfrac{\left( 1-\beta_{1} c+\alpha \gamma \beta_{1}
\right) ^{2}}{\gamma \beta_{1} \left( 1-\beta_{1} c\right) }\right) \allowbreak
\left( \beta_{1} \delta \left( c\beta_{1} -1\right) \right) +\left( \dfrac{%
\allowbreak \beta_{1} ^{2}\gamma \delta \left( c\beta_{1} -1\right) ^{2}}{\left(
\alpha \beta_{1} \gamma -c\beta_{1} +1\right) ^{2}}\right) >0$

$T_{c}=-u_{c}+\dfrac{\left( 1-\beta_{1} c+\alpha \gamma \beta_{1} \right) ^{2}}{
\gamma \beta_{1} \left( 1-\beta_{1} c\right) }+\allowbreak \beta_{1} \delta \left(
c\beta -1\right) $

We can thus state the following lemma

\begin{lemma}
\label{prcs}
Consider a parameter set s.t. $T_{nc} > 0$, that is the interior equilibrium point $\left(u_{nc},v_{nc}\right)$ of the no cannibalism model \eqref{eq:1} is unstable.
If $ \gamma <  \dfrac{ \beta }{\beta_{1}} $ then there exists a cannibalism rate $c$ s.t $T_{c} < 0$. This implies that the unstable interior equilibrium point $\left(u_{nc},v_{nc}\right)$ of the no cannibalism model \eqref{eq:1} can be stabilised via predator cannibalism. However If $ \gamma > \dfrac{ \beta }{\beta_{1}} $ then the unstable interior equilibrium point $\left(u_{nc},v_{nc}\right)$ of the no cannibalism model \eqref{eq:1} cannot be stabilised via predator cannibalism.
\end{lemma}

\begin{proof}
We take a simple geometric approach and visualise the nullclines below in fig. \ref{fig:nc1}. In order to shift the nullcline to the right due to predator cannibalism, we will require that the slope of the predator nullcline is less than $\beta$. That is $ \dfrac{\gamma \beta_{1} }{1-\beta_{1} c} < \beta$. If $ \dfrac{\gamma \beta_{1} }{\beta} = 1-\epsilon < 1$, we can always choose $c$ small enough s.t $1 - \epsilon < 1-\beta_{1}c$. In this case the slope of the predator nullcline with predator cannibalism will be less than $\beta$. If on the other hand $ \dfrac{\gamma \beta_{1} }{\beta} < 1$, this is clearly an impossibility as $ \dfrac{\gamma \beta_{1} }{\beta} > 1 > 1-\beta_{1} c$.

\begin{figure}[!htp]
\begin{center}
\includegraphics[scale=0.37]{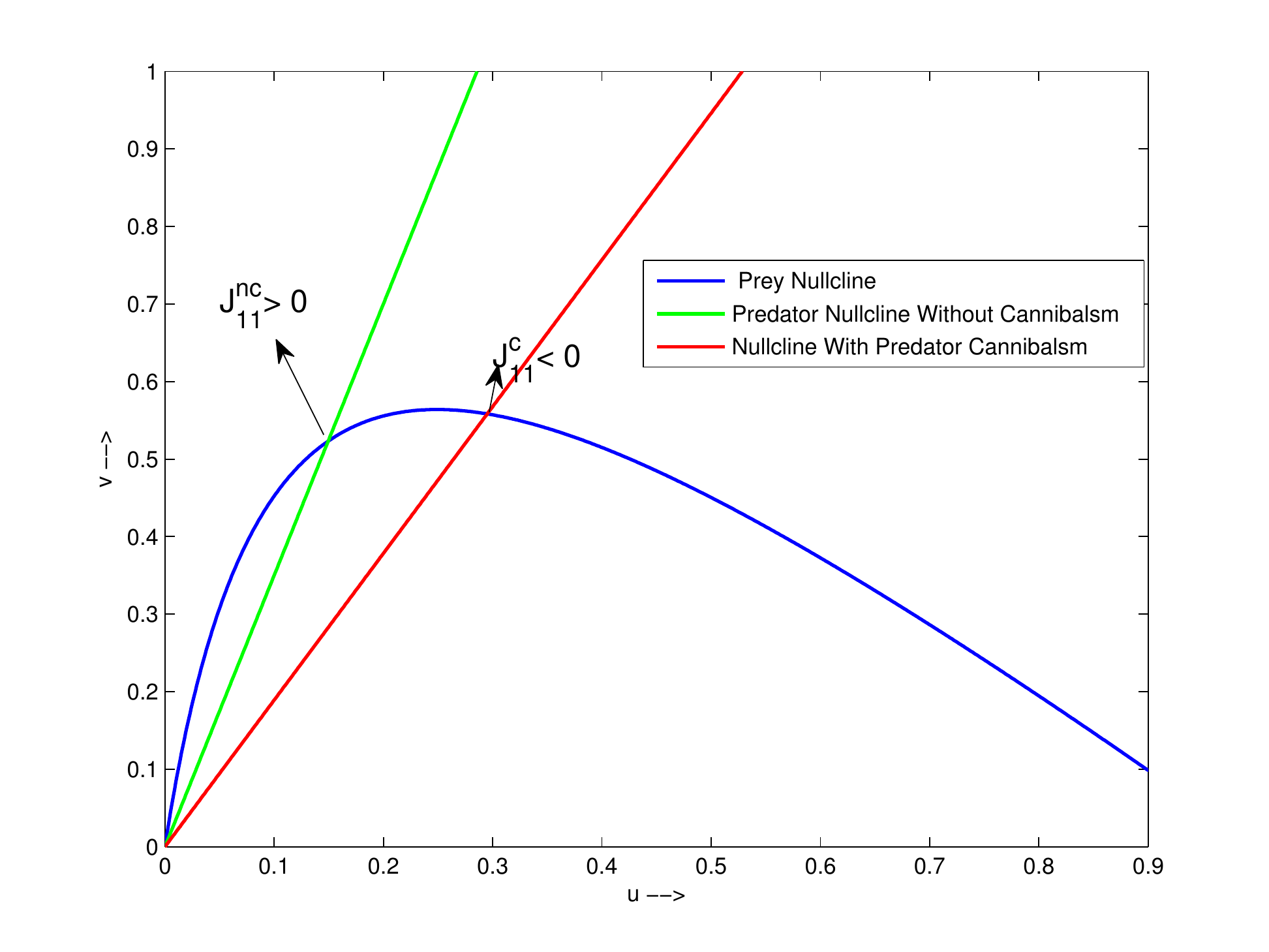} 
\end{center}
\caption{ Here we see the nullclines of the model \eqref{eq:1} vs
the nullclines of model \eqref{eq:prc}. We see that if $ \gamma <  \dfrac{ \beta }{\beta_{1}}$ predator cannibalism, decreases the slope of the 
 predator nullcline, shifting it to the right, and making it possible for $J^{c}_{11} <0$ and hence for $T_{c}<0$, while $T_{nc}>0$. In this figure the parameters chosen are
 $\alpha=0.89, \beta=3.5, c=0.18, \delta=0.01, \gamma=0.2$. }
\label{fig:nc1}
\end{figure}
\end{proof}

\begin{remark}
The cannibals feeding rate of prey $\gamma$ and the ratio $\dfrac{ \beta }{\beta_{1}} $, are critical in terms of stability. The stability condition tells us that if the cannibal eats less prey so $\gamma$ is small in comparison to $\dfrac{ \beta }{\beta_{1}} $, then only a little bit of feeding of conspecifics is enough to stabilize the dynamics.
\end{remark}

\subsection{The Holling-Tanner model with prey cannibalism}
We now consider the inclusion of prey cannibalism in \eqref{eq:1}. This yields the following model

\begin{equation}
 \label{eq:8}
\left. 
\begin{array}{c}
\dfrac{du}{dt}=u\left( 1+c_{1}-u\right) -\dfrac{uv}{u+\alpha v}-  c \dfrac{u^2}{u+d}  \\ 
\dfrac{dv}{dt}=\delta v\left( \beta -\dfrac{v}{u}\right)
\end{array}
\right\} prey \ cannibalism \ model 
\end{equation}

In deriving the above model we make the following assumptions 

\begin{itemize}
 \item The prey $u$ is depredating on its own species $u$.  
 
 \item Here the generic cannibalism term $C(u)$, is added in the prey equation, and is given by

 \begin{equation*}
C(u)= c \times u\times \dfrac{u}{u+d}.
\end{equation*}

That is the functional response of the cannibalistic prey is of Holling type II, with cannibalism rate $c$.

 \item There is a clear gain of energy to the cannibalistic prey, from the act of cannibalism. This gain results in an increase in reproduction in the prey. 
 This in turn leads to a gain in prey population, modeled via adding a $c_{1}u$ term to the prey equation. 

 \item We restrict $c_{1} < c$, as it takes depredation of a number of prey by the cannibal to produce one new offspring. 
\end{itemize}

The non-trivial steady state $\left( u_{c},v_{c}\right)$ is given by

\begin{equation*}
u_{c}=\dfrac{-\left( m+d+c-c_{1}-1\right) +\sqrt{\left( m+d+c-c_{1}-1\right)
^{2}+4d(1-m)+4dc_{1}}}{2}  \label{12}
\end{equation*}

where $m=\dfrac{\beta }{\alpha \beta +1} < 1$.

The feasibility conditions are divided into the following cases,

\begin{enumerate}
\item if $m<1+c_{1}$ and $m+d+c\geq 1+c_{1}$

\item if $m<1+c_{1}$ and $m+d+c<1+c_{1}$
\end{enumerate}

The Jacobian matrix of system  \eqref{eq:8} at the point $\left(
u_{c},v_{c}\right) $ given by:

\begin{center}
$J^{c}=\allowbreak \left( 
\begin{array}{ccc}
\left( \dfrac{\beta }{\left( 1+\alpha \beta \right) ^{2}}-u_{c}-\dfrac{cdu_{c}}{%
\left( u_{c}+d\right) ^{2}}\right) &  & -\dfrac{1}{\left( 1+\alpha \beta
\right) ^{2}} \\ 
&  &  \\ 
\beta ^{2}\delta &  & -\beta \delta%
\end{array}%
\right) $
\end{center}


The conditions for stability are:

\begin{equation*}
\mbox{Trace} \  J^{c}<0,  \mbox{det} \  J^{c}>0 \ 
\end{equation*}
The Trace $J^{c}<0$ if
\begin{equation*} 
\left( \dfrac{\beta }{\left( 1+\alpha \beta \right) ^{2}}-u_{c}-\dfrac{cdu_{c}}{
\left( u_{c}+d\right) ^{2}}\right) <\beta \delta 
\end{equation*}
The above implies
\begin{equation*} 
 \left( \dfrac{\beta }{\left( 1+\alpha \beta \right) ^{2}
}-u_{c}-\beta \delta \right) \dfrac{\left( u_{c}+d\right) ^{2}}{du_{c}}<c
\end{equation*} 
Note
 \begin{equation*}
 \mbox{Det} \  J^{c} = \dfrac{ \beta \delta u_{c} [d^{2}+2du_{c}+cd+(u_{c})^{2}]}{d^{2}+2du_{c}+(u_{c})^{2}} >0
 \end{equation*}
  This is always true.

\subsection{Effect of prey cannibalism on stability}
We now ask the following question. Consider the nc-model in a parameter regime, such that the interior equilibrium $(u_{nc}, v_{nc})$ is stable. Can prey cannibalism destabilize this equilibrium? If so, under what conditions?
This is settled via the following lemma,

\begin{lemma}
For any selection of parameters, there always exists a cannibalism rate $c$\ \ s.t the
interior equilibrium point $\left( u_{nc},v_{nc}\right)$ of the $nc$-model is stable, while the interior equilibrium point $\left(u_{c},v_{c}\right) $ of the prey cannibalism
model is unstable.
\end{lemma}

\begin{proof}

Note that the condition under which the
interior equilibrium point $\left( u_{nc},v_{nc}\right)$ of the $nc$-model is stable, while the interior equilibrium point $\left(u_{c},v_{c}\right) $ of the prey cannibalism
model is unstable, is 

\begin{equation*}
\left( T_{nc}<0<T_{c}\right) 
\end{equation*}

The above is equivalent to showing that $\exists c$ \ s.t \ 

\begin{equation*}
u_{c}+\dfrac{cdu_{c}}{\left( u_{c}+d\right) ^{2}}
<u_{nc}.
\end{equation*} 
 
 Note
 
 \begin{equation*}
u_{c}+\dfrac{cdu_{c}}{\left( u_{c}+d\right) ^{2}}
< u_{c}+\dfrac{c\frac{\left( u_{c}+d\right) ^{2}}{2}}{\left( u_{c}+d\right) ^{2}} = u_{c}+\dfrac{c}{2} < u_{nc}.
\end{equation*} 
 
 Thus we want to show 
 
\begin{equation*}
u_{c}+\dfrac{c}{2} < u_{nc}.
\end{equation*} 
 
Now $u_c=\dfrac{-\left( m+d+c-c_{1}-1\right) +\sqrt{\left( m+d+c-c_{1}-1\right)
^{2}+4d(1-m)+4dc_{1}}}{2}$ 

It is easily computable that  

\begin{equation*}
 \dfrac{d\left(
u_{c}\right) }{dc} =   -\dfrac{1}{2}+\dfrac{\left(
m+d+c-1-c_{1}\right) }{2\sqrt{\left( m+d+c-1-c_{1}\right) ^{2}+4d(1-m)+4dc_{1}}} <0 ,
\end{equation*} 

that is $u_{c}$ is monotonically decreasing in $c$, and so decreasing $c$ to a small enough level, will eventually cause $u_{c}+\dfrac{c}{2} < u_{nc}$.
This proves the lemma.
\end{proof}

We next ask the following question. Consider the nc-model in a parameter regime, such that the interior equilibrium $(u_{nc}, v_{nc})$ is unstable. Can prey cannibalism stabilize this equilibrium? If so, under what conditions?
In order to answer this question, we must further analyse the dynamical behavior of the prey cannibalism model. We state our result via the following lemma:

\begin{lemma}
\label{pcns}
If \ $\left( m+d+c\right) <1+c_{1}$\ then then there does not exist a cannibalism rate $c$\ \ s.t $%
u_{c}+\dfrac{cdu_{c}}{\left( u_{c}+d\right) ^{2}}>u_{nc}$. This implies that no amount of 
prey cannibalism can stabilize the unstable interior equilibrium point $\left(u_{nc},v_{nc}\right)$ of the prey cannibalism
model.
\end{lemma}

\begin{proof}
The explicit condition under which the equilibrium
point $\left( u_{nc},v_{nc}\right)$ of the nc-model is unstable and
 the equilibrium point $\left( u_{c},v_{c}\right)$ 
 of the prey cannibalism model is stable is given by

\begin{equation*}
\left( T_{c}<0<T_{nc}\right) 
\end{equation*}

This is equivalent to $u_{c}+\dfrac{cdu_{c}}{\left( u_{c}+d\right) ^{2}} < u_{nc}$.
Now,
\begin{equation*}
u_{c}+\dfrac{cdu_{c}}{\left( u_{c}+d\right) ^{2}} = u_{c}+\dfrac{c}{2}\dfrac{2du_{c}}{\left( u_{c}+d\right) ^{2}} < u_{c}+
\dfrac{c}{2}\dfrac{\left( u_{c}+d\right) ^{2}}{\left( u_{c}+d\right) ^{2}}=u_{c}+\dfrac{c}{2}.
\end{equation*}

Thus if we can show 

$u_{c}+\dfrac{c}{2}<u_{nc}$

then this will imply that 
$u_{c}+\dfrac{cdu_{c}}{\left( u_{c}+d\right) ^{2}}<u_{c}+\dfrac{c}{2}<u_{nc}$.

We let $f(c)=u_{c}+\dfrac{c}{2}-u_{nc}$, note $f(0)=0$.

We next compute the derivative of $f$ w.r.t. $c$.

\begin{eqnarray}
&&f(c)  \nonumber \\
&& = \dfrac{-\left(m+d+c-1-c_{1}\right) +\sqrt{\left( m+d+c-1-c_{1}\right)
^{2}+4d(1-m)+4dc_{1}}}{2}+\dfrac{c}{2}-u_{nc} \nonumber \\
&& \Rightarrow f^{^{\prime }}(c) = -\dfrac{1}{2}+\dfrac{\left(
m+d+c-1-c_{1}\right) }{2\sqrt{\left( m+d+c-1-c_{1}\right) ^{2}+4d(1-m)+4dc_{1}}}  + \dfrac{1}{2}\nonumber \\
&& =\dfrac{\left( m+d+c-1-c_{1}\right) }{2\sqrt{\left(
m+d+c-1-c_{1}\right)^{2}+4d(1-m)+4dc_{1}}} \nonumber \\
&& < 0  \nonumber \\
\end{eqnarray}

since $\left( m+d+c\right) <1+c_{1}$. This implies that $f^{^{\prime }}(c) < 0$. Thus $f(c)$ is decreasing, and since $f(0) = 0$, it must be that $f(c) = u_{c}+\dfrac{c}{2}-u_{nc} < 0$, for all $c$. Thus
$u_{c}+\dfrac{cdu_{c}}{\left( u_{c}+d\right) ^{2}} < u_{c}+\dfrac{c}{2} < u_{nc}$.
This proves the lemma.
\end{proof}

%
%

\begin{remark}
The above proof is under the first feasibility condition $\left( m+d+c\right) <1+c_{1}$, that is small cannibalism parameter $c$.
Numerically, we could not find a cannibalism rate $c$ such that the equilibrium point $\left( u_{nc},v_{nc}\right)$ 
of the nc-model is unstable while the equilibrium point $\left(u_{c},v_{c}\right)$ of the prey cannibalism model is stable, under the second feasibility condition  $\left( m+d+c\right) >1+c_{1}$, that is when the  cannibalism parameter $c$ can be very large.
We tried a parameter sweep of all parameters in the interval $[0,10]$. We thus beleive that prey cannibalism cannot stabilise, whatever the range of the cannibalism parameter $c$ be, large or small. However, the large $c$ case is unproven at this point.
\end{remark}

Nonetheless, we make the following conjecture

\begin{conjecture}
There does not exist a cannibalism rate $c$, such that the equilibrium point $\left( u_{nc},v_{nc}\right)$ 
of the nc-model is unstable while the equilibrium point $\left(u_{c},v_{c}\right)$ of the prey cannibalism model is stable. 
\end{conjecture}

\section{The effect of {\emph{Noisy}}  cannibalism rate on limit cycle dynamics}
In this section we aim to investigate the effect of a noisy cannibalistic parameter on the limit cycle dynamics in the prey cannibalism model.
We proceed by stating the following lemma,

\begin{lemma}
\label{bound}
The solutions of system \eqref{eq:8} are always positive and bounded, furthermore
there exists $T\geq 0$ such that $0<u(t)<1+c_{1},$ $0<v(t)<\beta(1+c_{1}),$ for all $t > T$.
\end{lemma}

\begin{proof}
The positivity easily follows by the quasi monotonicity of the reaction terms.
For boundedness note that for system \eqref{eq:8}, we have 

$\frac{du}{dt} \leq u(1+c_{1}-u)$

A standard comparison argument shows that $\lim_{t\rightarrow \infty}u(t)\leq 1+c_{1}.$It follows that there exists $T>0$ such that $u(t)<1+c_{1},$ for
$t>T$.

From the second equation of system \eqref{eq:8}, we see that, for $t>T$,


A standard comparison argument shows that  $\lim_{t\rightarrow \infty
}v(t)\leq (1+c_{1})\beta$.
\end{proof}

Now, it can be shown that under certain parametric restrictions, there exists a limit cycle in the prey cannibalism model  \eqref{eq:8}.
We state this via the following lemma,

\begin{lemma}
\label{lc}
Consider the prey cannibalism model \eqref{eq:8}.
Assume that the following conditions hold

\begin{equation}
\alpha \beta + 1>\beta , \  \left( \dfrac{\beta }{\left( 1+\alpha \beta \right) ^{2}
}-u_{c}-\beta \delta \right) \dfrac{\left( u_{c}+d\right) ^{2}}{du_{c}}>c,
\end{equation}

then the system possesses at least one limit cycle.
\end{lemma}
Here $u_{c}$ is the prey equilibrium in model \eqref{eq:8}.
\begin{proof}
From the local stability analysis, we have that if the above conditions hold, then the positive equilibrium $E_{c}=\left(
u_{c},v_{c}\right) $ of system  \eqref{eq:8} is an unstable node or focus. The
existence of a limit cycle now follows from lemma \ref{bound} and Poincair\'{e}%
-Bendixson theorem \cite{strogatz2014nonlinear}.
\end{proof}

\begin{remark}
Although the existence of a unique limit cycle for the nc-model is well known, and the existence of a limit cycle with the addition of prey cannibalism
is seen via lemma \ref{lc}, we are unable to prove uniqueness, in this case. We turn our attention to investigating the effect that a stochastic cannibalism rate has on the limit cycle dynamics.
\end{remark}
Biological species are always subject to random environmental fluctuations, be it via climatic factors, or ecological. Also, many times the parameters we use for modeling such processes many not be known with certainty.
In this present paper, we only assume that the cannibalism rate $c$
is affected by environmental fluctuations, which is define as mortality rate due
to cannibalism. We introduce a stochastic factor in 
time ${\mathrm{\eta}(t)}$, into the cannibalism rate $%
(c)$, and so the parameter $c$ is treated as a Gaussian random variable with mean $\hat{c}$ and intensity $\alpha_{1}$. Here ${\mathrm{\eta}(t)}=``\frac{d}{dt}W_{t}"$ is a classical white noise. 
\begin{align} 
 \label{equ:NoisyCoefficient}
c \approx \hat{c} + \alpha_{1} {\mathrm{\eta}(t) }.
\end{align}

Under these assumptions we consider the following SDE model, where for1simplicity, we redefine  mean $\hat{c}$ to be  the parameter $c$.

\begin{eqnarray}
\label{eq:8s}
&&du=\left(u\left( 1+c_{1}-u\right) -\dfrac{uv}{u+\alpha v}- c \dfrac{u^2}{u+d}\right)dt  + \alpha_{1}\dfrac{u^2}{u+d} dW_{t}  \nonumber \\
&&dv=\delta v\left( \beta -\dfrac{v}{u}\right)dt \nonumber \\
\end{eqnarray}

 We observe that as $\alpha_{1}$ is increased, the limit cycle dynamics becomes aperiodic, but some cyclicity remains. The noise is unable to stabilise, that is drive the population cycle to a steady state. This is shown via the following plots,

\begin{figure}[!htp]
\begin{center}
\includegraphics[scale=0.26]{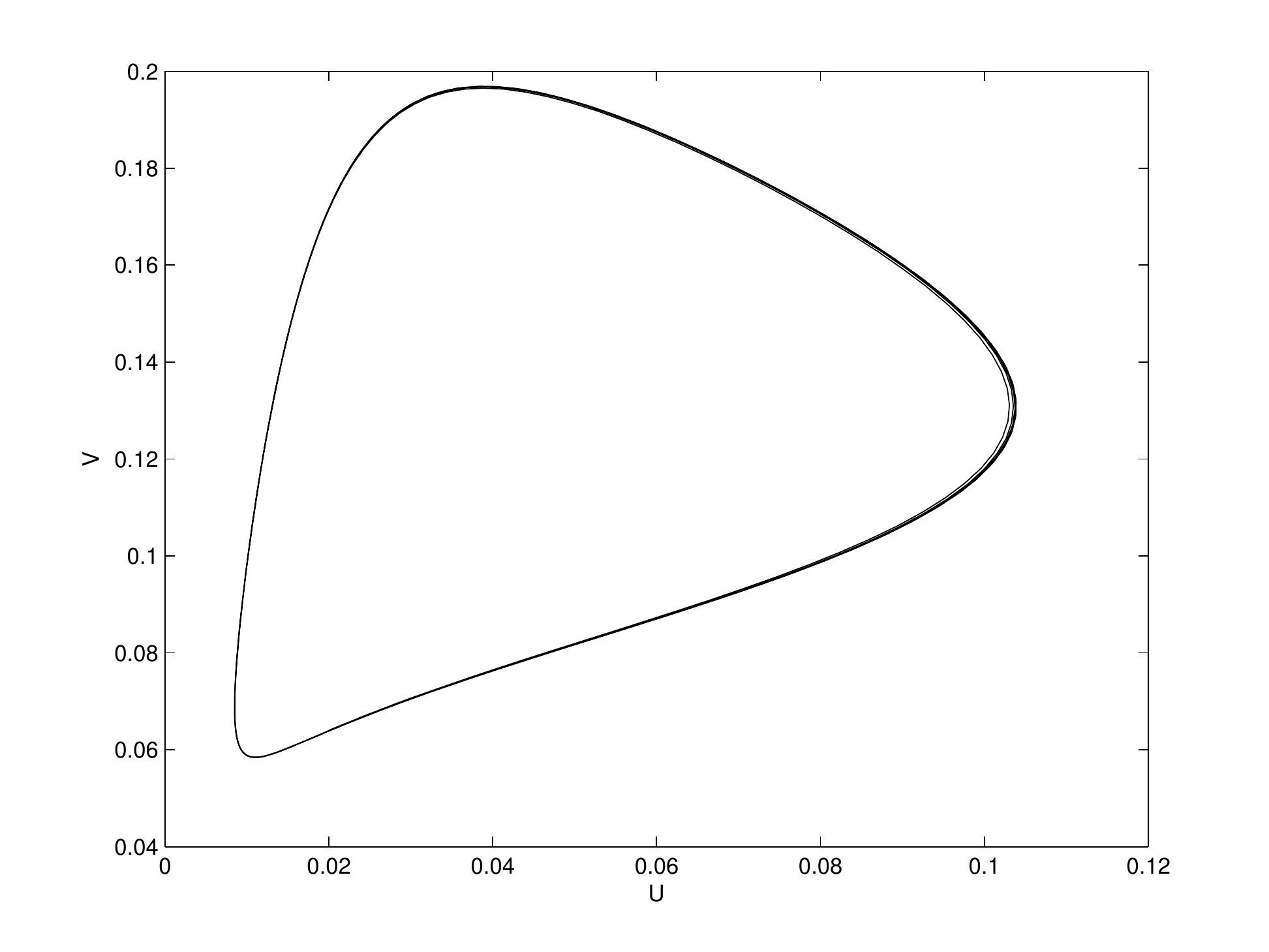}  %
\includegraphics[scale=0.26]{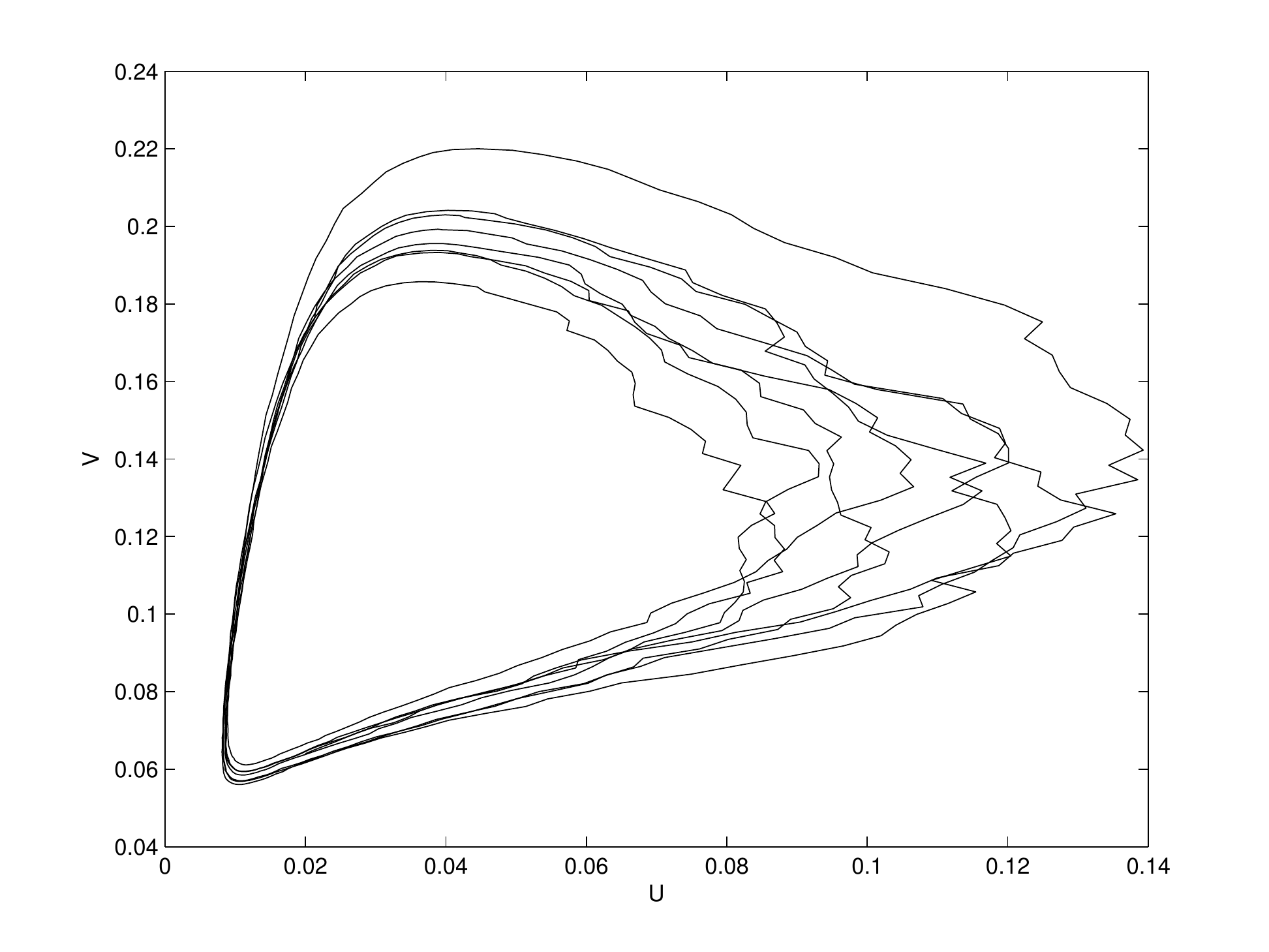} 
\end{center}
\caption{Here we see a limit cycle in nc-model in the left panel (with $\alpha_1=0.0$). We observe that the effect of a noisy prey cannibalism rate is to distort the limit cycle, but some cyclicity remains (with $\alpha_1=0.1$). For the list of parameters chosen please refer to Table \ref{table:Paramset}. }
\label{fig:2s}
\end{figure}

\begin{figure}[!htb]
\begin{center}
\subfigure{
\includegraphics[scale=0.26]{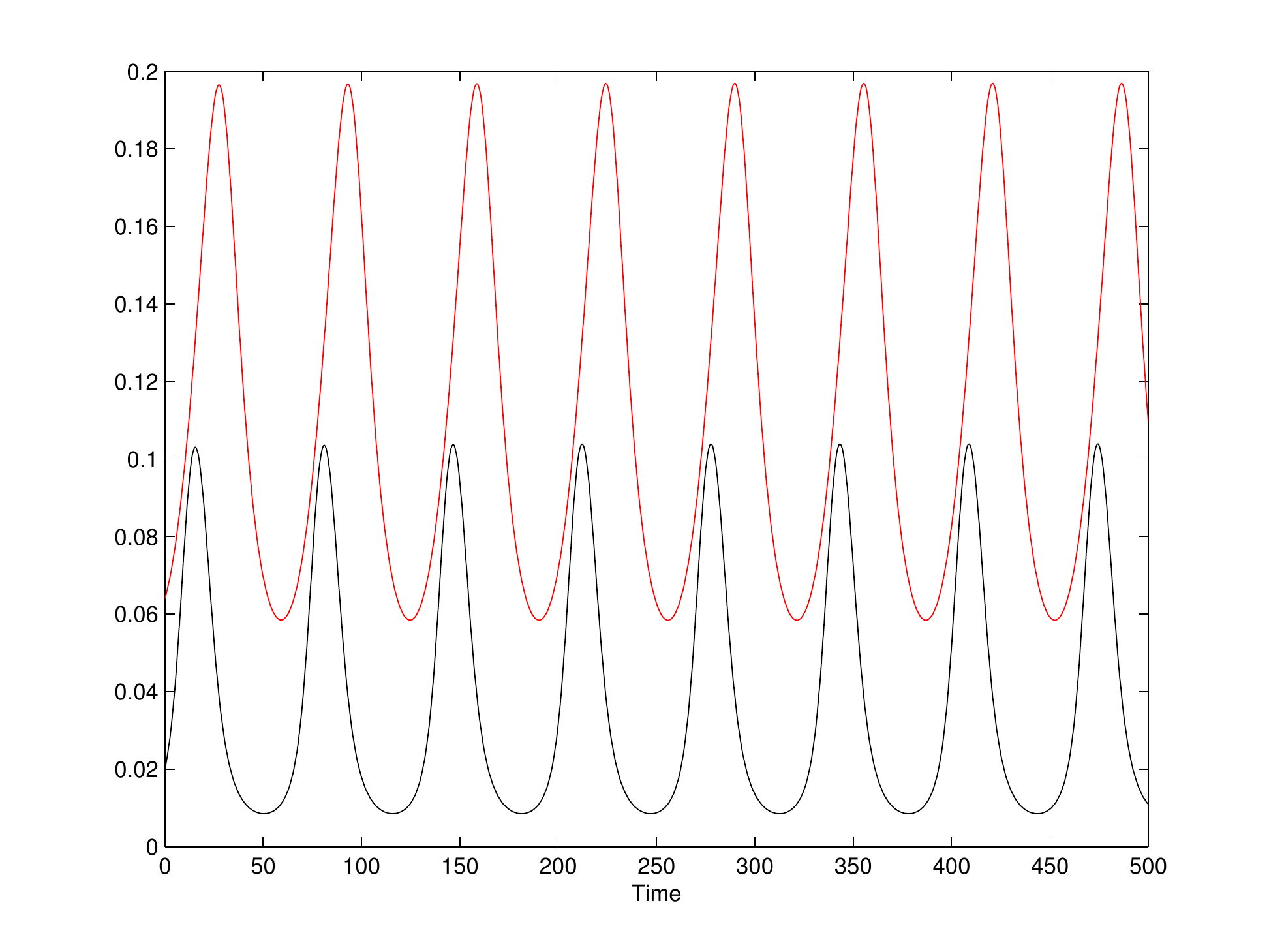}} 
\subfigure{
\includegraphics[scale=0.26]{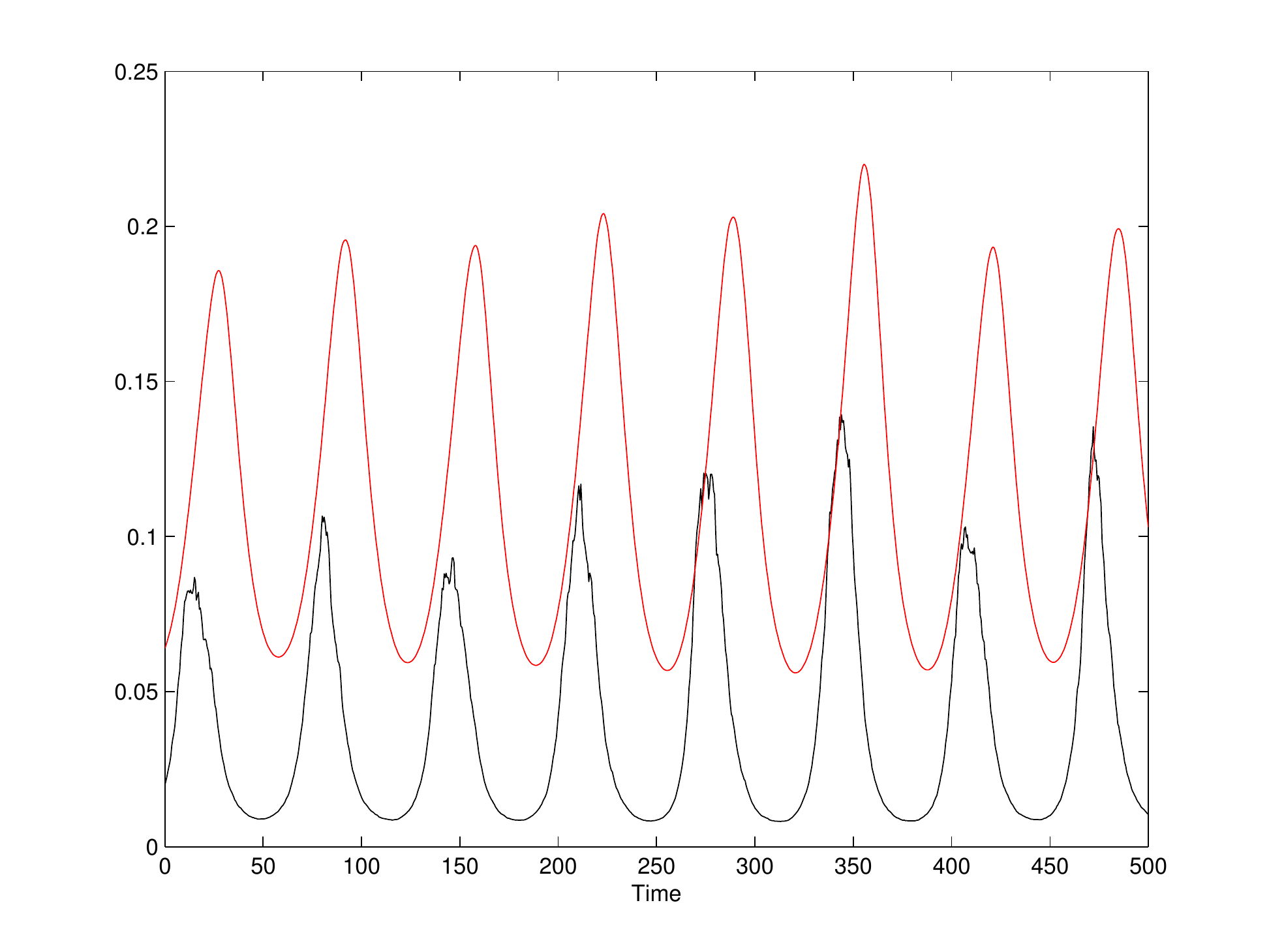}}
\end{center}
\caption{The left panel in the above is the time series corresponding to the phase plot in the left panel in fig. \ref{fig:2s}  (with $\alpha_1=0.0$). The right panel corresponds to a slightly greater noise intensity (with $\alpha_1=0.1$). For the list of parameters chosen please refer to Table \ref{table:Paramset}. }
\label{fig:3s}
\end{figure}

\begin{figure}[!htp]
\begin{center}
\includegraphics[scale=0.26]{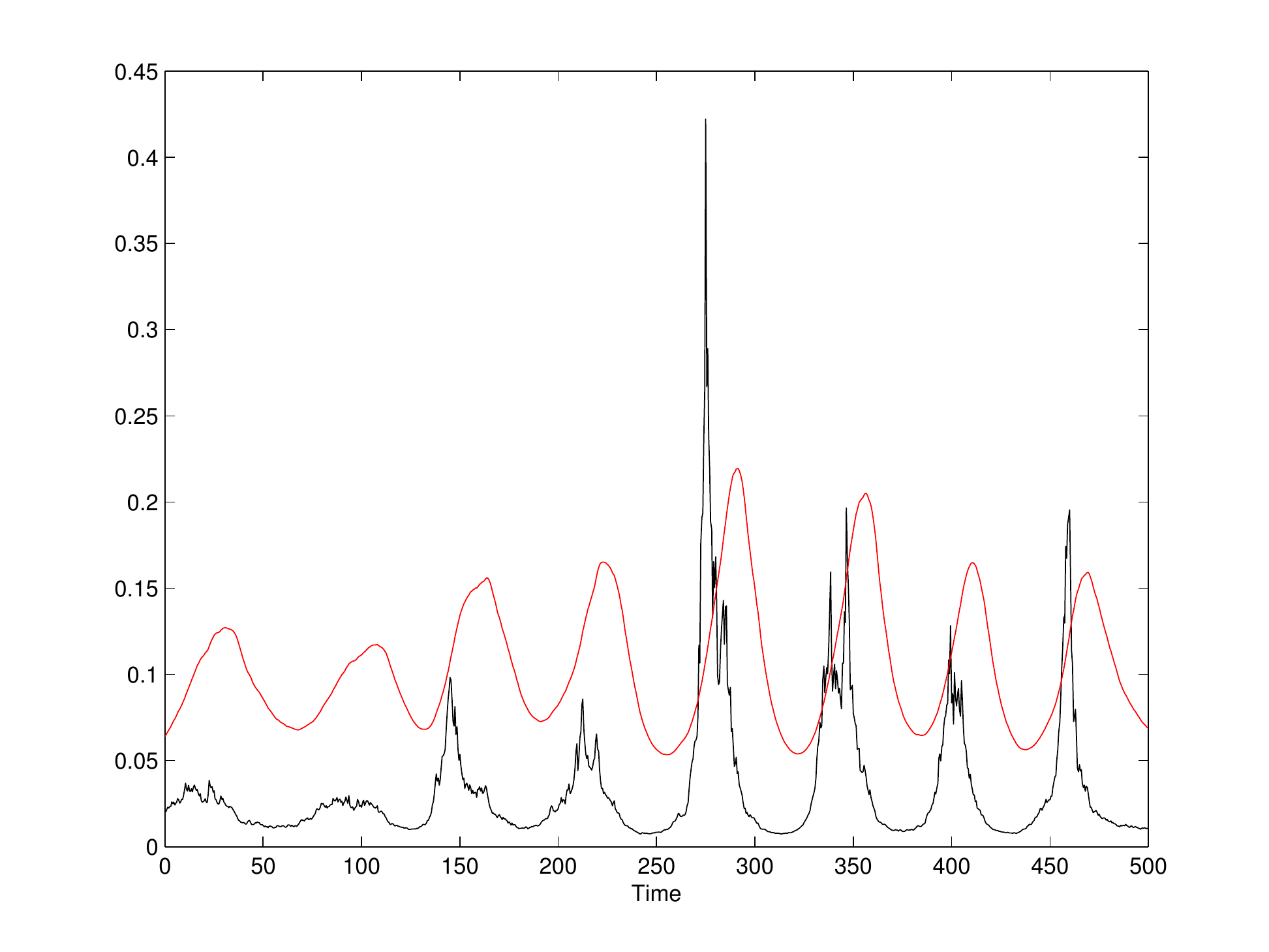}  %
\includegraphics[scale=0.26]{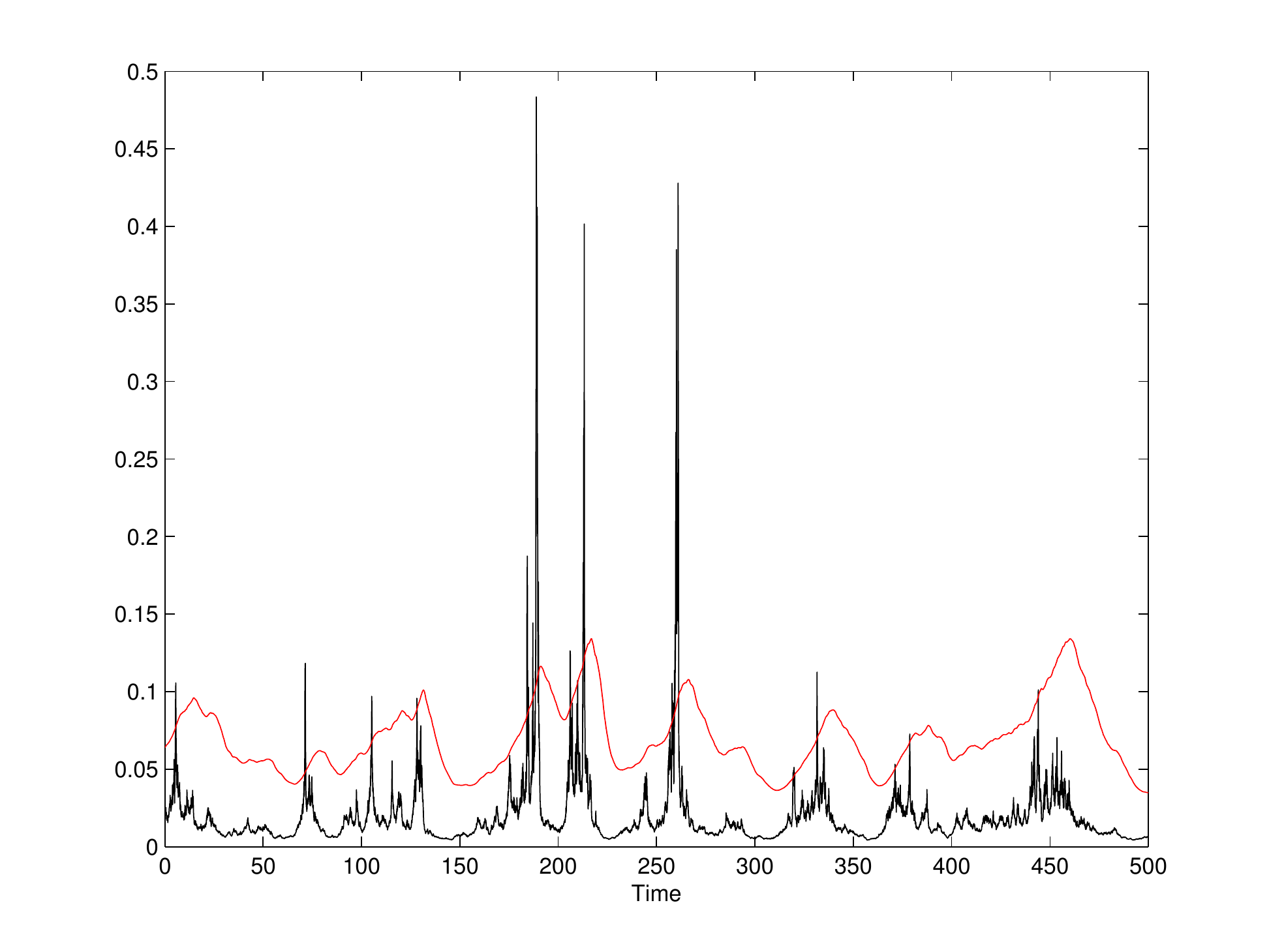} 
\end{center}
\caption{These figures show the effect of further increasing the noise intensity, on the limit cycle dynamics (where the left panel shows $\alpha_1=0.5$ and right panel shows  $\alpha_1=2.0$). For the list of parameters chosen please refer to Table  \ref{table:Paramset}. }
\label{fig:4s}
\end{figure}

\begin{remark}
Thus it seems that even with some stochasticity in the cannibalism parameter, prey cannibalism is unable to stabilise dynamics. It is well known that large enough noise intensity will drive everything to extinction, but this is an uninteresting case.
\end{remark}

This leads us to make the following conjecture

\begin{conjecture}
There does not exist a stochastic cannibalism rate $c$, such that the non trivial equilibrium point $\left( u_{nc},v_{nc}\right)$ 
of the nc-model is unstable, while the nontrivial equilibrium point $\left(u_{c},v_{c}\right)$ of the stochastic prey cannibalism model is stable. 
\end{conjecture}

\section{The effect of prey cannibalism on Turing instability}

Turing instability occurs when the positive interior equilibrium point for
an ODE system is stable in the absence of diffusion, but becomes unstable
due to the addition of diffusion, that is in the PDE model. This phenomena is also referred to as
diffusion driven instability, \cite{Turing52}. 
To formulate the spatially explicit form of the earlier described ODE models, we assume that the predator and prey populations move actively in space. Random movement of animals occurs because of various requirements and necessities like, search for better food, better opportunity for social interactions such as finding mates \cite{okubo2001diffusion}. Food availability and living conditions demand that these animals migrate to other spatial locations. In the proposed model, we have included diffusion assuming that the animal movements are uniformly distributed in all directions. 
The models with diffusion are described as follows
\begin{equation}
 \label{eq:MBmodel_1}
\left. 
\begin{array}{c}
u_{t}=D_u u_{xx} + u( 1-u) -{\frac{uv}{u+\alpha v}}, \\
v_{t}=D_v v_{xx} + \delta v\left( \beta -\frac{v}{u}\right)
\end{array}
\right\}  \ no \ cannibalism \ model 
\end{equation}

and the diffusive model with prey cannibalism

\begin{equation}
 \label{eq:MBmodel_New}
\left. 
\begin{array}{c}
u_{t}=D_u u_{xx} + (1+c_{1})u( 1-u) -{\frac{uv}{u+\alpha v}} - \frac{cu^2}{u+d}, \\
v_{t}=D_v v_{xx} + \delta v\left( \beta -\frac{v}{u}\right)
\end{array}
\right\} \ prey \ cannibalism \ model 
\end{equation}

The spatial domain for the above is $\Omega \subset \mathbb{R}^{1}$. $\Omega$
is assumed bounded, and we prescribe Neumann boundary conditions $u_{x} =
v_{x} = 0$, for $x \in \partial \Omega$, and suitable initial conditions. 
The main reason for choosing such boundary conditions is that we are interested in the self-organization of patterns, 
and zero-flux boundary conditions imply that no external input is imposed \cite{okubo2001diffusion, M93}. In the numerical simulations we take $\Omega = [0,\pi]$.
 Here $u(x,t)$ is the density of the prey species at any given time $t$, that is logistically controlled, and predated
on by a predator species with density $v(x,t)$. 

\begin{remark}
For various models, that fall into the above general category, it has been
shown that Turing instability exists. However, this is only for certain
parameters that lie in the so called Turing space. For parameters outside of
that space Turing patterns are an impossibility. 
\end{remark}

\subsection{General criterion for no Turing instability in nc-model}
Our goal in the current
section is to investigate the conditions under which there is no Turing instability in \eqref{eq:MBmodel_1}, but there are Turing patterns in \eqref{eq:MBmodel_New}. Thus we aim to quantify the conditions under which cannibalism \emph{in the prey species},
can induce Turing patterns.

We begin by linearizing model system \eqref{eq:MBmodel_1} about its homogenous
steady state  $(u_{nc},v_{nc})$, by introducing both space and time-dependent fluctuations. In order to derive the precise conditions for Turing instability there is a well
known procedure. We refer the reader to \cite{Gilligan1998}. 
The calculations therein can be stated via the following theorem,

\begin{theorem}[Turing Instability Condition]
\label{Thm312} Consider $(u_{nc},v_{nc})$ which is a spatially homogenous
equilibrium point of the model system \eqref{eq:MBmodel_1}. For a given set of
parameters, if at $(u_{nc},v_{nc})$ , the Jacobian  $\mathbf{J} = 
\begin{bmatrix}
J_{11} & J_{12} \\ 
J_{21} & J_{22}%
\end{bmatrix}%
$, of the reaction terms, and the diffusion coefficients $D_u, D_v$ satisfy 

\begin{align*}
& J_{11} +J_{22}<0 \\
& J_{11}J_{22} - J_{21}J_{12}>0 , \\
& D_v J_{11} + D_uJ_{22}>0. \\
& {(D_v J_{11} + D_uJ_{22})^2} -4D_uD_v( J_{11}J_{22} - J_{21}J_{12})>0, \\
\end{align*}
then $(u_{nc},v_{nc})$ is said to be linearly stable in the {\emph{a}bsence of
diffusion} and linearly unstable in the {\emph{p}resence of diffusion}.
\end{theorem}

And easy corollary derived from the above imposes an additional necessary
condition for Turing instability

\begin{corollary}[Necessary condition for Turing Instability]
\label{cor312} If $(u_{nc},v_{nc})$ is a spatially homogenous equilibrium point of
the model system \eqref{eq:MBmodel_1} which is linearly stable in the {\emph{a}%
bsence of diffusion}, and linearly unstable in the {\emph{p}resence of
diffusion}, then either $J_{11} < 0, J_{22} > 0$ or $J_{11} > 0, J_{22} < 0$%
, where $\mathbf{J} = 
\begin{bmatrix}
J_{11} & J_{12} \\ 
J_{21} & J_{22}%
\end{bmatrix}%
$, is the Jacobian matrix of the reaction terms.
\end{corollary}

We begin by focusing on \eqref{eq:MBmodel_1}, investigated in
detail by Banerjee \cite{banerjee2012turing}. 
The results therein reveal that turing instability in model %
\eqref{eq:MBmodel_1} occurs, if $\alpha < 1$, for a range of parameters in the so
called Turing space.

Note that the jacobin matrix of the reaction terms of \eqref{eq:MBmodel_1} at the point $\left(u_{nc
},v_{nc}\right)$ is given by

\begin{center}
$J=$ $\allowbreak \allowbreak \left( 
\begin{array}{ccc}
\dfrac{\beta }{\alpha \beta +1}+\dfrac{\beta }{\left( \alpha \beta +1\right)
^{2}}-1 &  & -\dfrac{1}{\left( \alpha \beta +1\right) ^{2}} \\ 
&  &  \\ 
\beta ^{2}\delta &  & -\beta \delta%
\end{array}%
\right) $
\end{center}

Since the parameters $\beta > 0, \delta >0$, we always have that $J_{22} < 0$. Thus an application of corollary \ref{cor312} yields the following lemma:
\begin{lemma}
\label{lemnt}  The equilibrium point $\left( u_{nc},v_{nc}\right) $
cannot be driven unstable via diffusion if $J_{11} < 0$, or equivalently if $%
\beta \left( \alpha \beta +2\right) <\left( \alpha \beta +1\right) ^{2}$.
\end{lemma}

It is a matter of simple algebra to convert the above. That is:

\begin{equation}  \label{eq:c12}
\beta \left( \alpha \beta +2\right) <\left( \alpha \beta +1\right) ^{2} \iff
\beta \leq \frac{1}{(\sqrt{1-\alpha})(1+ \sqrt{1-\alpha})}
\end{equation}

\begin{remark}
Let us investigate this ratio further. Note, $\beta$ is the
intrinsic growth rate of the predator, and $\alpha= \frac{\mbox{prey \ handling \ time}} {\mbox{ prey \ capture \ rate }} $. (These are
nondimensionalised parameters, and due to feasibility constraints $\alpha < 1
$). Roughly speaking, \eqref{eq:c12} is saying that if

\begin{equation*}
{\mbox{predator \ growth \ rate}  \leq \frac{1}{\sqrt{1- \left(\frac{\mbox{prey \ handling \ time}} {\mbox{ prey \ capture \ rate\ }}\right)}}}
\end{equation*}

then Turing instability is an impossibility in the nc model \eqref{eq:MBmodel_1}. Now, given a predator
growth rate $\beta$, If ``prey handling time" is close enough to the ``prey capture rate",  then 
\eqref{eq:c12} will hold, and Turing instability cannot occur in 
\eqref{eq:MBmodel_1}. 
\end{remark}

We now pause and ask the following question: Can
cannibalism in the prey species alter this scenario? Specifically, can prey
cannibalism, bring about Turing instability in \eqref{eq:MBmodel_New}, even if \eqref{eq:c12} holds?
To answer this question, we investigate the Turing instability conditions for \eqref{eq:MBmodel_New}. The calculations therein yield that for the spatially homogenous equilibrium
solution $\left(u_{c},v_{c}\right)$ to \eqref{eq:MBmodel_New} the
Jacobian matrix at $\left(u_{c},v_{c}\right) $ is given by:

$J=\allowbreak \left( 
\begin{array}{ccc}
\left( \dfrac{\beta }{\left( 1+\alpha \beta \right) ^{2}}-u_{c}-\dfrac{%
cdu_{c}}{\left( u_{c}+d\right) ^{2}}\right) &  & -\dfrac{1}{\left(
1+\alpha \beta \right) ^{2}} \\ 
&  &  \\ 
\beta ^{2}\delta &  & -\beta \delta%
\end{array}%
\right) $

For cannibalism in the prey to alter Turing dynamics via lemma \ref{lemnt},
at the very least we require  that $J_{11} > 0$. Clearly there is no Turing,
if $J_{11} < 0$. Now, if $J_{11} > 0$ it does not guarantee there is Turing, as this condition is only necessary,
but \emph{there is a chance for Turing instability to occur}. Thus In order for there to be a possibility of Turing instability in \eqref{eq:MBmodel_1}, while there is no 
Turing instability in \eqref{eq:MBmodel_New}, we would require $J^{c}_{11}>0$, that is 

\begin{equation}
\left( \dfrac{\beta }{\left( 1+\alpha \beta \right) ^{2}}-u_{c}
\right)  \dfrac{\left( u_{c}+d\right) ^{2}}{u_{c}}>cd, \ \mbox{whilst} \  \beta < \frac{1}{(\sqrt{1-\alpha})(1+ \sqrt{1-\alpha})}.
\end{equation}

We can also take a geometric approach similar to Fasani \cite%
{fasani2012remarks} to proceed. We investigate the nullclines of %
\eqref{eq:MBmodel_1}, shown in fig. \ref{fig:nc}.

\begin{figure}[!htp]
\begin{center}
\includegraphics[scale=0.37]{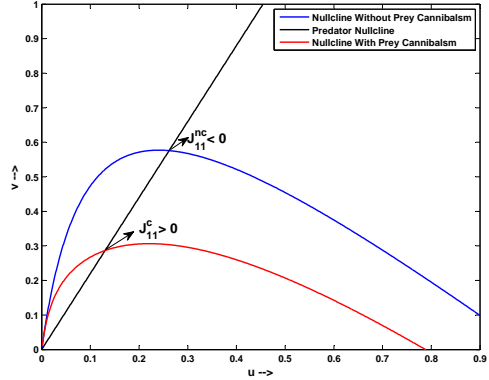} 
\end{center}
\caption{ Here we see the nullclines of the model \eqref{eq:MBmodel_1} vs
the nullclines of model \eqref{eq:MBmodel_New}. We see prey cannibalism
drives down the prey nullcline.}
\label{fig:nc}
\end{figure}

\begin{remark}
Since the situation in fig. \ref{fig:nc} is clearly achievable geometrically, there is a
chance that prey cannibalism could bring about Turing instability, even if
$ \beta < \frac{1}{(\sqrt{1-\alpha})(1+ \sqrt{1-\alpha})}$. We then perform a parameter sweep
based on the Turing conditions that we derive, in order to see if the
appropriate parameters can be found.
\end{remark}
\subsection{Turing instability due to prey cannibalism}
For details on the Turing calculations we refer the reader to \cite{Gilligan1998}. The calculations therein yield various necessary and
sufficient conditions for Turing instability to occur, if prey cannibalism
is present. We state our result via the following theorem.
\begin{theorem}[ Prey cannibalism induced Turing Instability]
\label{Thm3123} Consider $(u_{c},v_{c})$ which is a spatially homogenous
equilibrium point of the model system \eqref{eq:MBmodel_New}. If there
exist parameters $(\alpha,\beta,c,d,c_{1})$ all strictly positive, s.t at $%
(u_{c},v_{c})$, the Jacobian  $\mathbf{J} = 
\begin{bmatrix}
J_{11} & J_{12} \\ 
J_{21} & J_{22}%
\end{bmatrix}%
$, and the diffusion coefficients $D_u, D_v$ satisfy 
\begin{align*}
&\left( \dfrac{\beta }{\left( 1+\alpha \beta \right) ^{2}}-u_{c}\right) 
\dfrac{\left( u_{c} +d\right) ^{2}}{du_{c}}>{c}%
>\left( \dfrac{\beta }{\left( 1+\alpha \beta \right) ^{2}}-u_{c}- \delta
\beta \right) \dfrac{\left( u_{c}+d\right) ^{2}}{du_{c}} \\
& \left( D_{v}\left( \frac{\beta }{\left( 1+\alpha \beta \right) ^{2}}%
-u_{c}-\frac{cdu_{c}}{\left( u_{c}+d\right) ^{2}}\right)
-D_{u}\beta \delta \right) >0 \\
& \left( D_{v}\left( \frac{\beta }{\left( 1+\alpha \beta \right) ^{2}}%
-u_{c}-\frac{cdu_{c}}{\left( u_{c}+d\right) ^{2}}\right)
-D_{u}\beta \delta \right) >2\sqrt{D_{u}D_{v}\left( \tfrac{\left( \beta
\delta d^{2}u_{c}+2\beta \delta d(u_{c})^{2}+c\beta \delta du_{c}+\beta \delta (u_{c})^{3}\right) \allowbreak }{d^{2}+2du_{c}+(u_{c})^{2}}\right) }, \\
&\beta \leq \frac{1}{(\sqrt{1-\alpha})(1+ \sqrt{1-\alpha})} , \\
\end{align*}
then $(u_{c},v_{c})$ is linearly stable in the {\emph{a}bsence of diffusion} and
linearly unstable in the {\emph{p}resence of diffusion}, whilst $%
(u_{nc},v_{nc})$ which is a spatially homogenous equilibrium point of the
model system \eqref{eq:MBmodel_1} is linearly stable in the {\emph{p}%
resence of diffusion}, for the same parameters $(\alpha,\beta,d)$, and $c=c_{1}=0$.
 That is, in the absence of prey cannibalism.
\end{theorem}

\begin{remark}
Note for system \eqref{eq:MBmodel_New} and for \eqref{eq:MBmodel_1} we have
that $J_{11}J_{22} - J_{21}J_{12}>0 $.
\end{remark}

A parameter sweep using the above conditions yields that such a parametric space is definitely possible. For chosen parameters in this space, we see Turing patterns as in fig. \ref{fig:2}.

\begin{figure}[!htp]
\begin{center}
\includegraphics[scale=0.26]{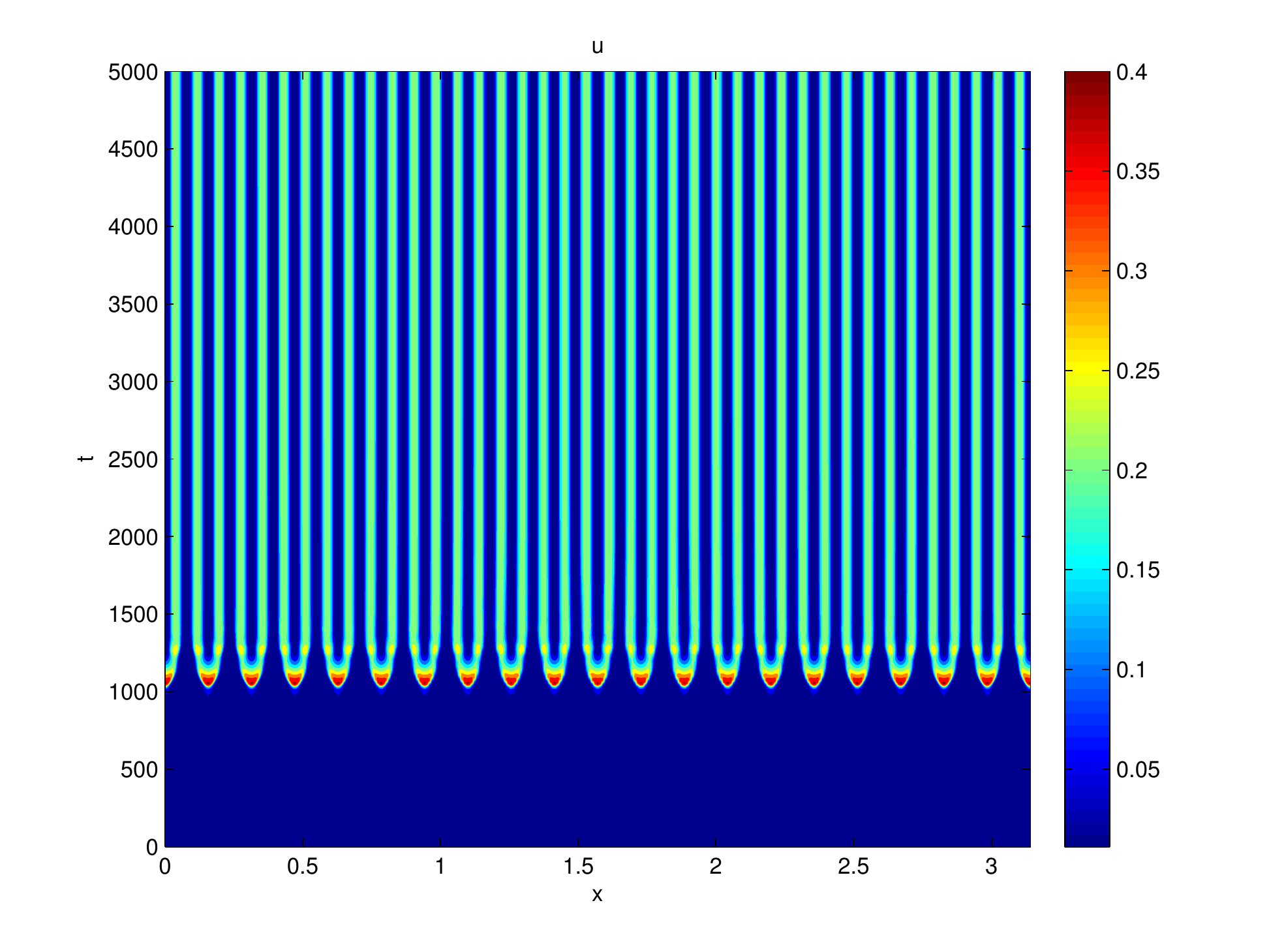}  %
\includegraphics[scale=0.26]{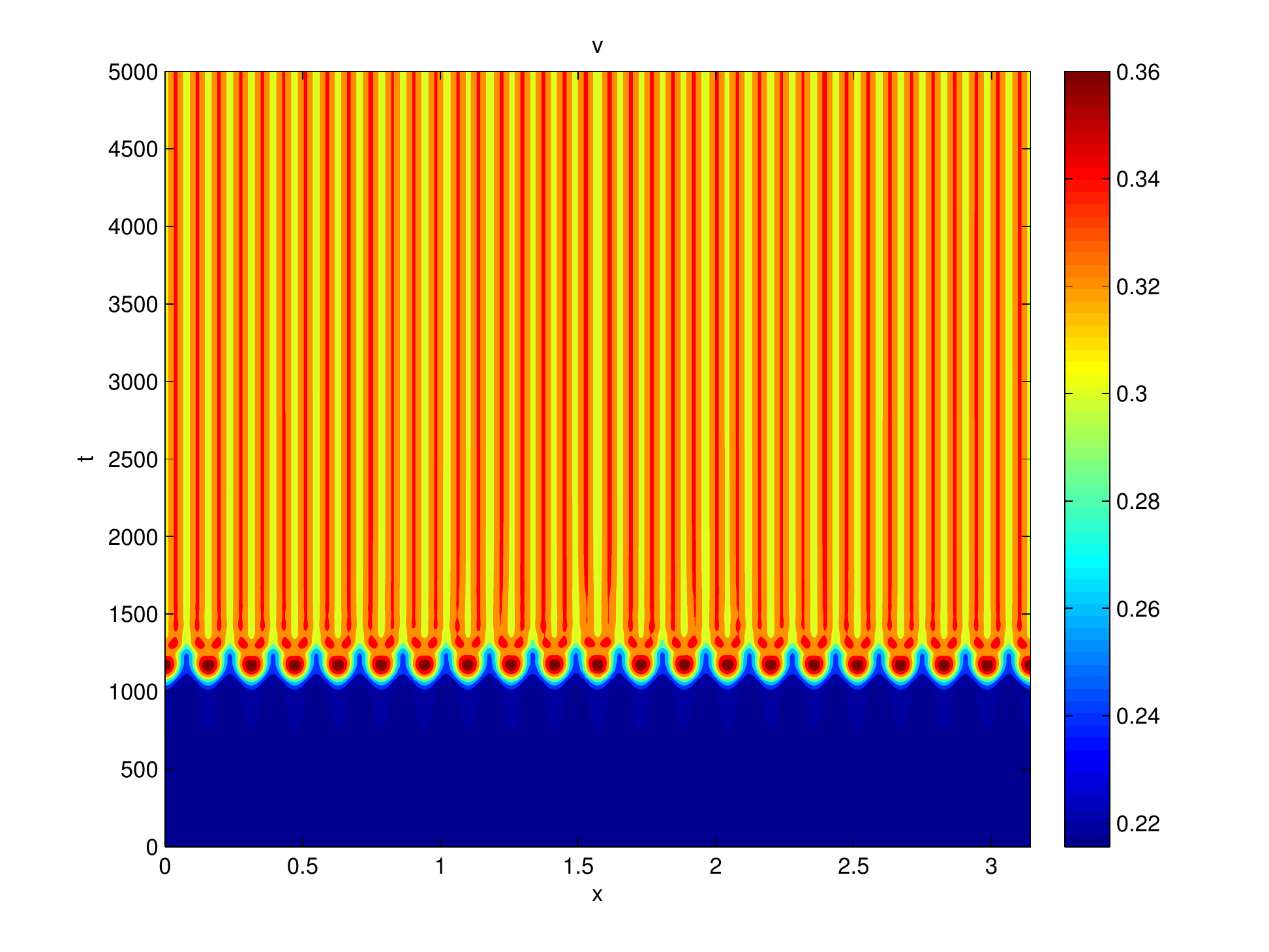} 
\end{center}
\caption{Here the Turing patterns in predator and prey, are shown for model 
\eqref{eq:MBmodel_New}, that is when there is prey cannibalism. Spatial
patterns are formed. In these simulations $\Omega =[0,\pi]$, and the spatially homogenous steady state is given a small perturbation by a function of the form $\epsilon_{1} cos^2(10x)$, $\epsilon_{1} <<1$. For the parameters used please see Table \protect\ref%
{table:Paramset}.}
\label{fig:2}
\end{figure}

\begin{figure}[!htb]
\begin{center}
\subfigure[]{
\includegraphics[scale=0.26]{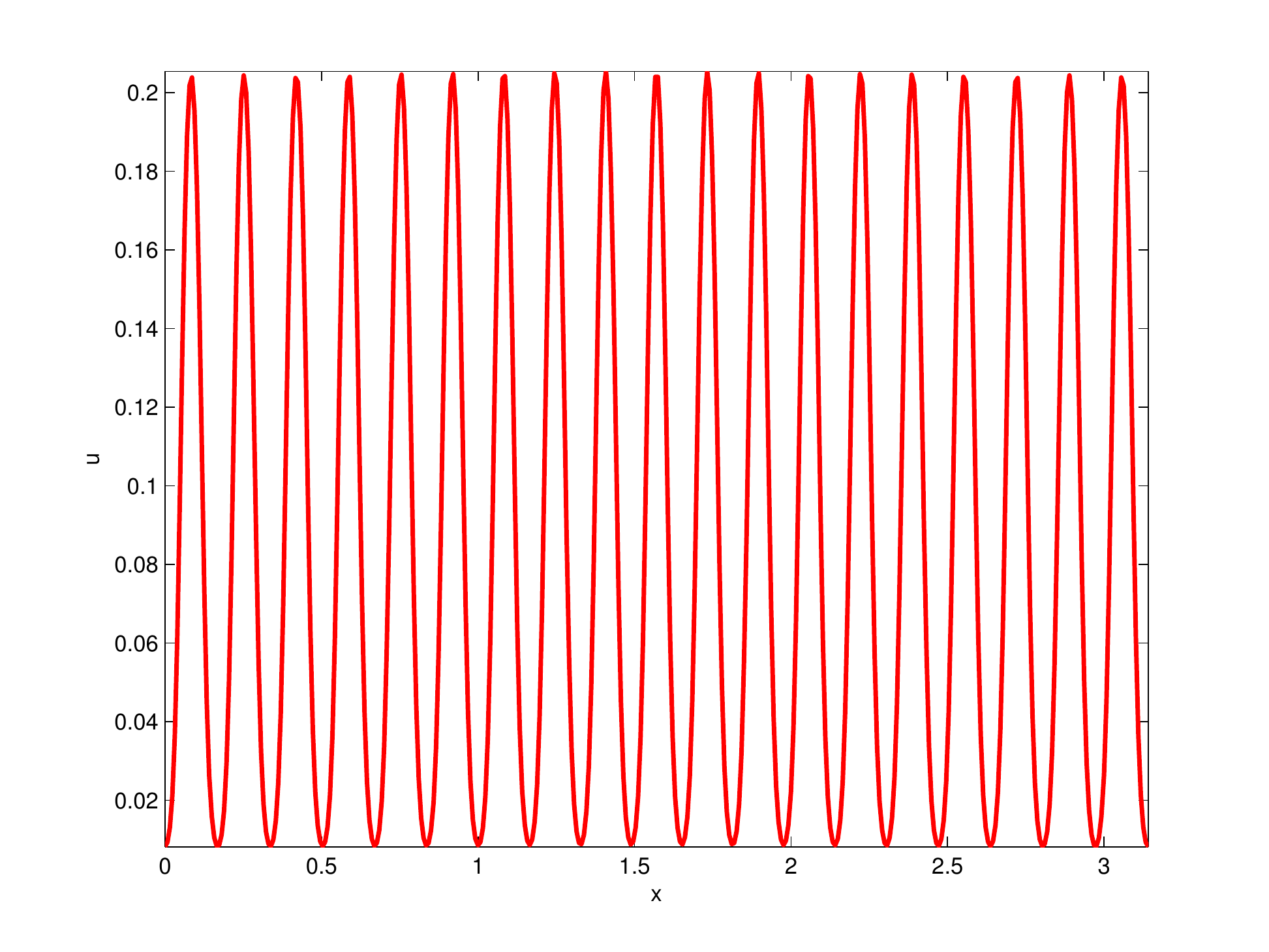}} 
\subfigure[]{
\includegraphics[scale=0.26]{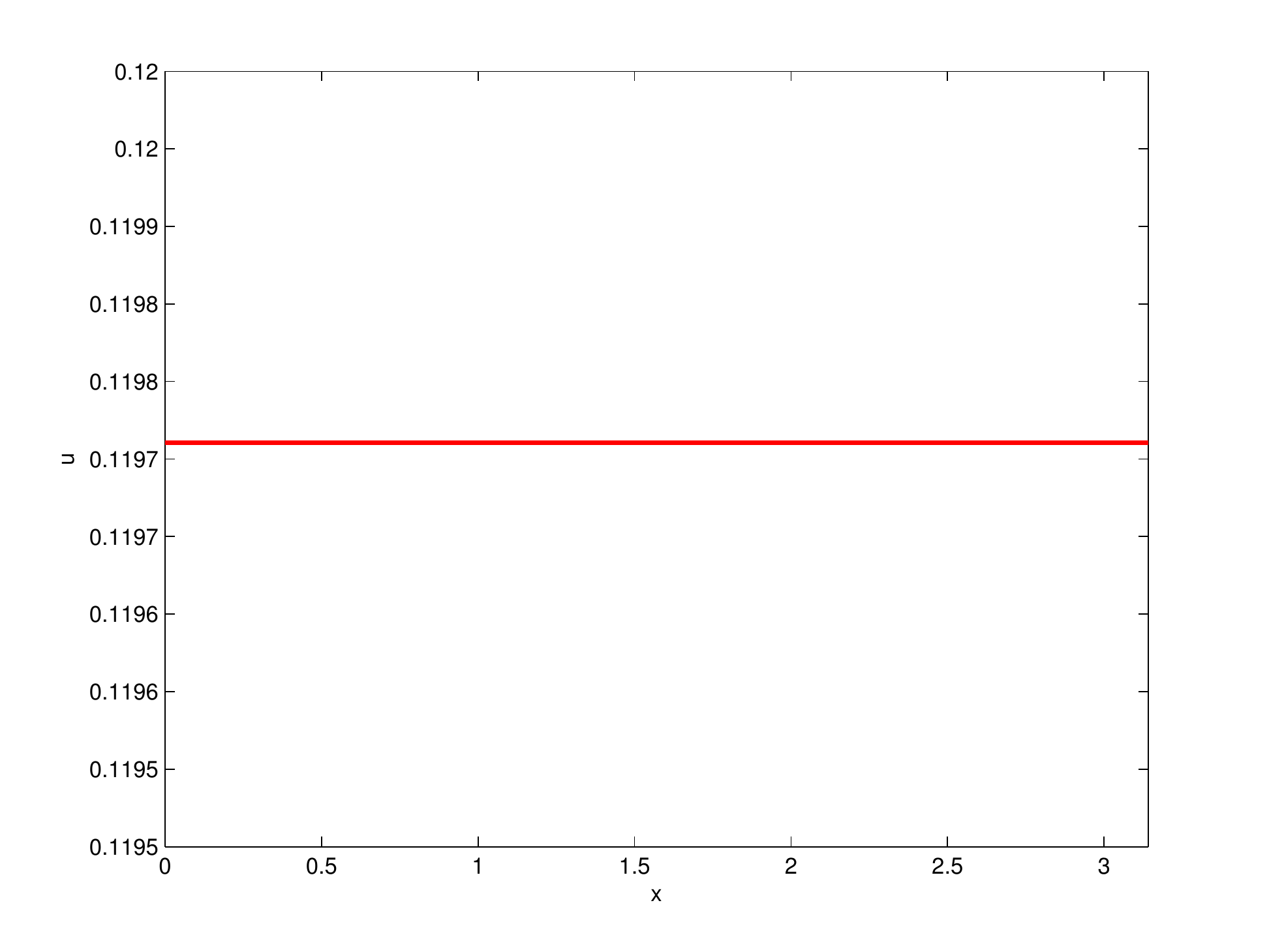}}
\end{center}
\caption{Here the model \eqref{eq:MBmodel_New} (that is when there is prey
cannibalism), is contrasted with model \eqref{eq:MBmodel_1}. We see in Fig
(a) (model \eqref{eq:MBmodel_New}) the density of the prey at time t=5000.
There is a spatial pattern. In Fig (b) (model \eqref{eq:MBmodel_1}) there
is no pattern. For the parameters used please see Table \protect\ref%
{table:Paramset}.}
\label{fig:3}
\end{figure}

Now we move in the opposite direction. Supposing that the conditions of
theorem \ref{Thm3123} do not hold, and we do not have Turing instability in
the nc model \eqref{eq:MBmodel_New}. Can removing the prey cannibalism
bring about Turing instability? We next provide a sufficient condition under
which this is not possible. This is done via the following lemma

\begin{lemma}
\label{lem3123} Consider $(u^{c},v^{c})$ which is a spatially homogenous
equilibrium point of the model system \eqref{eq:MBmodel_New} which is
linearly stable both in the {\emph{a}bsence of diffusion} and in the {\emph{p%
}resence of diffusion}. If $J^{c}_{11}<0$ or $\left( \left( \dfrac{\beta }{\left( 1+\alpha
\beta \right) ^{2}}-u^{c }\right) \left( u^{c }+d\right) ^{2}\right) 
\dfrac{1}{u^{c}}<cd$, and if $m+d+c\leq 1+c_{1}$, and $1-\frac{\beta}{1+\alpha \beta} < \frac{1}{2+\alpha \beta}$. Then $(u^{nc},v^{nc})$ which is a spatially
homogenous equilibrium point of the model system \eqref{eq:MBmodel_1}, for
the same parameter set, but when $c=c_{1}=0$, is also linearly stable in the {\emph{p%
}resence of diffusion}
\end{lemma}

\begin{proof}
The proof follows via observing that if 

\begin{eqnarray}
&&\left( \left( \dfrac{\beta }{\left(
1+\alpha \beta \right) ^{2}}-u^{c}\right) \left( u^{c }+d\right)
^{2}\right) \dfrac{1}{u^{c }}>cd \nonumber \\
&& =>  \dfrac{\beta }{\left(1+\alpha \beta \right) ^{2}} < u^{c} + \frac{c d u^{c}}{ \left( u^{c }+d \right)^{2}} < u_{nc} \nonumber \\
&&  = 1-\frac{\beta}{1+\alpha \beta} < \frac{1}{2+\alpha \beta} \nonumber \\
\end{eqnarray}

This implies that $\beta(\alpha \beta + 2) < (1+\alpha \beta)^{2}$, which in turn implies that 
 $\beta < \frac{1}{(\sqrt{1-\alpha})(1+ \sqrt{1-\alpha})} $, and so Turing instability is an impossibility in \eqref{eq:MBmodel_1}, from lemma \ref{lemnt}. Note 
$u^{c} + \frac{c d u^{c}}{\left( u^{c }+d\right)^{2}} < u_{nc}$ follows from the assumption $m+d+c\leq 1+c_{1}$, and lemma \ref{pcns}.
\end{proof}
\begin{remark}
Essentially removing prey cannibalism, shifts the
prey null cline up, still leaving $J^{nc}_{11} < 0$, at the new equilibrium $%
(u^{nc},v^{nc})$. This is easily seen in fig. \ref{fig:nc}.
\end{remark}

We lastly state a result concerning Turing pattern formation and predator cannibalism

\begin{lemma}
\label{lemnt1}  
The equilibrium point $\left( u_{nc},v_{nc}\right) $
cannot be driven unstable via diffusion, in the presence of predator cannibalism, if  and $\gamma < \frac{\beta}{\beta_{1}}$.
However, if $\gamma > \frac{\beta}{\beta_{1}}$, then $\left( u_{nc},v_{nc}\right) $
can be driven unstable via diffusion, in the presence of predator cannibalism.
\end{lemma}

\begin{proof}
The proof is an immediate consequence of the null cline analysis. If $J^{nc}_{11} = \beta \leq \frac{1}{(\sqrt{1-\alpha})(1+ \sqrt{1-\alpha})} < 0$, Turing instability is an impossibility. However, predator cannibalism will only push the null cline further to the right if $\gamma < \frac{\beta}{\beta_{1}}$, and so $J^{c}_{11} <0$, thus Turing patterns are an impossibility, in this case. However, if $\gamma > \frac{\beta}{\beta_{1}}$, then predator cannibalism will push the predator null cline to the left, and it is possible that $J^{c}_{11} > 0$, so that Turing patterns could occur.
\end{proof}

\section{Dynamics of the PDE model with noise}

In mathematical modeling, deterministic systems of PDEs have been relied
upon to capture the relevant dynamics of some biological systems. However,
these systems are frequently subject to noisy environments, inputs and
signalling. It was reported that songbirds experience decreased predation
rate in the environments with noise \cite{francis2009noise}, and hermit
crabs were distracted by boat motor noise and hence less vigilant against
approaching predators \cite{chan2010anthropogenic,siemers2011hunting}. Thus
stochastic perturbations are important when considering the ability of such
models to reproduce results consistently. In particular, owing to the Turing
mechanism being very sensitive to perturbations, stochasticity is able to
affect evolving patterns in ways not seen in deterministic simulations \cite%
{maini2012turing}.

Previous result in stochastic predator-prey models focused on internal noise%
\cite{li2013dynamics,kelkel2010stochastic,sun2010rich}, where it is assumed
that noise affect the model system itself. This allows to account for such
fluctuations even in the framework of a continuum model \cite%
{kelkel2010stochastic}. However, predator-prey models with Holling type II
external noise had been generally overlooked despite its potential
ecological reality and theoretical interest. Recent theoretical work
indicates that cannibalism can strongly promote coexistence between a prey
and a predator, especially if the mortality rate due to cannibalism is
sufficiently larger than the heterospecific predation rate \cite%
{rudolf2007interaction}, as a result, we investigate the effect of external
spatial temporal noise on Turing patterns in a predator-prey model with
diffusion and Holling type II functional response, with prey cannibalism. We
model these perturbations as said, by means of independent Gaussian \emph{%
white noises}.

In this present paper, we only assume some control parameter of the system
is affected by environmental fluctuations such as epidemics, weather and
nature disasters. We consider fluctuations presented in by the coefficient
of the cannibalistic term of the prey, which is define as mortality rate due
to cannibalism. Hence by introducing a stochastic factor both in space and
time ${\mathrm{\eta}(x,t) }$ into the coefficient of the cannibalistic term $%
(c)$, the parameter $c$ is treated as a Guassian random variable with mean $\hat{c}$ and intensity $\epsilon$. 

\begin{align}  \label{equ:NoisyCoefficient}
c \rightarrow{\hat{c}} + \epsilon {\mathrm{\eta}(x,t) },
\end{align}
We therefore introduce equation \eqref{equ:NoisyCoefficient} in order to
extend model system \eqref{eq:MBmodel_New} such that we obtain a system of
It$\hat o^{\prime }s $ stochastic differential equations in the form, where for simplicity, we redefine  mean $\hat{c}$ to be  the parameter $c$. : 
\begin{subequations}
\label{equ:MB model}
\begin{align}
&{\frac{ \partial u }{\partial t}}=D_u\Delta u + u( 1+c_{1}-u) -{\frac{uv}{%
u+\alpha v}}-{\frac{cu^{2}}{u+d}} - {\frac{\epsilon u^{2}}{u+d}} {\mathrm{%
\eta}(x,t)}, \\
&{\frac{ \partial v }{\partial t}}=D_v\Delta v + \delta v\left( \beta -\frac{%
v}{u}\right).
\end{align}
\end{subequations}

\subsection{Numerical Experiment}

In this section, we discretize numerically the stochastic predator-prey
model \eqref{equ:MB model} with white noise. We redefine the noise term ${%
\mathrm{\eta}(x,t)}$ as ${\mathrm{\partial}^2 W(x,t)/ \partial x \partial t}$%
, where $W(x,t)$ is a one-dimensional Brownian motion. The spatial
approximation is constructed from a Chebychev collocation scheme\cite%
{gottlieb1983spectrum,don1995accuracy}. Thus the spatial approximation is
constructed as a linear combination of the interpolating splines on the
Gauss-Lobatto quadrature ${x_j}$, with a spatial domain size of N=256. This
reduces model system \eqref{eq:MBmodel_New} to a system of It$\hat o^{\prime
}s $ stochastic ordinary differential equation. The resulting system is then
integrated in time using an implicit Milsten scheme for approximating It$%
\hat o^{\prime }s $ stochastic ordinary differential equation, with a time
stepsize of $\Delta t=0.1$. This therefore reduce the It$\hat o^{\prime }s $
stochastic ordinary differential equation to the form 

\begin{align}\label{equ:MB modelNoise}
   \begin{split}
 {\bf{u}}^{n+1}_j -{\bf{u}}^{n}_j&=\bigg(D_u\sum_{i=1}^{N-1}{\bf{D^2}}_{ji}{\bf{u}}^{n+1}_i +  {\bf{u}}^{n+1}_j ( 1+c_{1}- {\bf{u}}^{n+1}_j ) -{ {\bf{u}}^{n+1}_j {\bf{v}}^{n+1}_j\over  {\bf{u}}^{n+1}_j +\alpha {\bf{v}}^{n+1}_j}- {c ({\bf{u}}^{n+1}_j) ^{2}\over  {\bf{u}}^{n+1}_j +d}\bigg)\Delta t  \\
    &\quad \quad \quad  \quad  \quad  \quad  \quad  \quad   \quad \quad
  - {\epsilon ({{\bf{u}}^{n}_j})^{2}\over {\bf{u}}^{n}_j+d} {\mathrm \Delta W_j}  - {\epsilon^2\over 2}{({{\bf{u}}^{n}_j})^3({\bf{u}}^{n}_j+2d)\over ({\bf{u}}^{n}_j+d)^3}\bigg( (\Delta W_j)^2 -\Delta t\bigg),
   \end{split}\\
 {\bf{v}}^{n+1}_j - {\bf{v}}^{n}_j &=\left(D_v\sum_{i=1}^{N-1}{\bf{D^2}}_{ji} {\bf{v}}^{n+1}_i + \delta {\bf{v}}^{n+1}_j\left( \beta -\frac{{\bf{v}}^{n+1}_j}{{\bf{u}}^{n+1}_j}\right)\right)\Delta t ,
\end{align}
where $\mathbf{D^2}_{ji}$ is known as the {\emph{2}nd} order spectral
differentiation matrix and $\mathbf{u}^{n}_{j}=u(x_j,t_n)$, whereas $\mathbf{%
v}^{n}_{j}=v(x_j,t_n)$, for $j=1,\ldots,N-1$. Also $\Delta W_j=z_j\sqrt{%
\Delta t}$, where $z_j\sim \mathcal{N}(0,1)$ is normally distributed random
numbers with mean 0 and standard deviation 1.\newline
The implicit Milsten scheme used to advance the initial value problem in
time necessitates solving a nonlinear system of equations at each time step.
This is achieved through a Newton-Raphson method.

\subsection{The effect of {\emph{Noisy}} cannibalism rate on Turing Instability}

In this section we investigate numerically by simulation the effects of {%
\emph{noise}} on Turing patterns at time $t=5000$. Extensive testing was
performed through numerical integration to describe model \eqref{equ:MB
model}, and the results are shown in this section. This is done by using an
initial condition defined in a small perturbation around the positive
homogenous steady state given, which is as 
\begin{align*}
u=u^{*} + \alpha_1 cos^2(nx)(x > 0)(x < \pi), \\
v=v^{*} + \alpha_2 cos^2(nx)(x > 0)(x < \pi), \\
\end{align*}
where $\alpha_1=0.0002$ and $\alpha_2=0.02.$ During the simulations,
different types of dynamics are obtained for both the prey and predator.
Consequently in this section, we can restrict our analysis of pattern
formation to prey population $u$ since much of the noise in model system %
\eqref{equ:MB model} affect the prey population.\newline
In our numerical simulation with a specific parameter set in the Turing
domain as show in Fig.\eqref{fig:TuringIntensity}, we show the evolution of
fixed spatial pattern of prey population from deterministic model (model
system \eqref{eq:MBmodel_New}) to stochastic model (model system %
\eqref{equ:MB model}) where with noise we consider $\epsilon = 0.1$ and $0.4$
with small random perturbation of the stationary solution $(u^*,v^*)$ as
given above. It is numerically evidenced from Fig.\eqref{fig:TuringIntensity} (first plot $\epsilon=0$, middle plot $\epsilon=0.1$ and last plot $\epsilon=0.4$)
that large variety of distinct patterns can be obtained by making small
changes to the level of intensity $\epsilon$ in the stochastic model, which
can also cause the disappears of Turing patterns as well as the pattern
formed can be enhance by small intensity levels. Therefore in our
experiments it was realized that for a given parameter set in the Turing
domain, there exist a specific intensity $\epsilon$ level of the noise that
can cause the average prey population species to decrease significantly as
well as the disappearance of pattern formation.\newline
\begin{figure}[!htp]
\begin{center}
\includegraphics[scale=0.17]{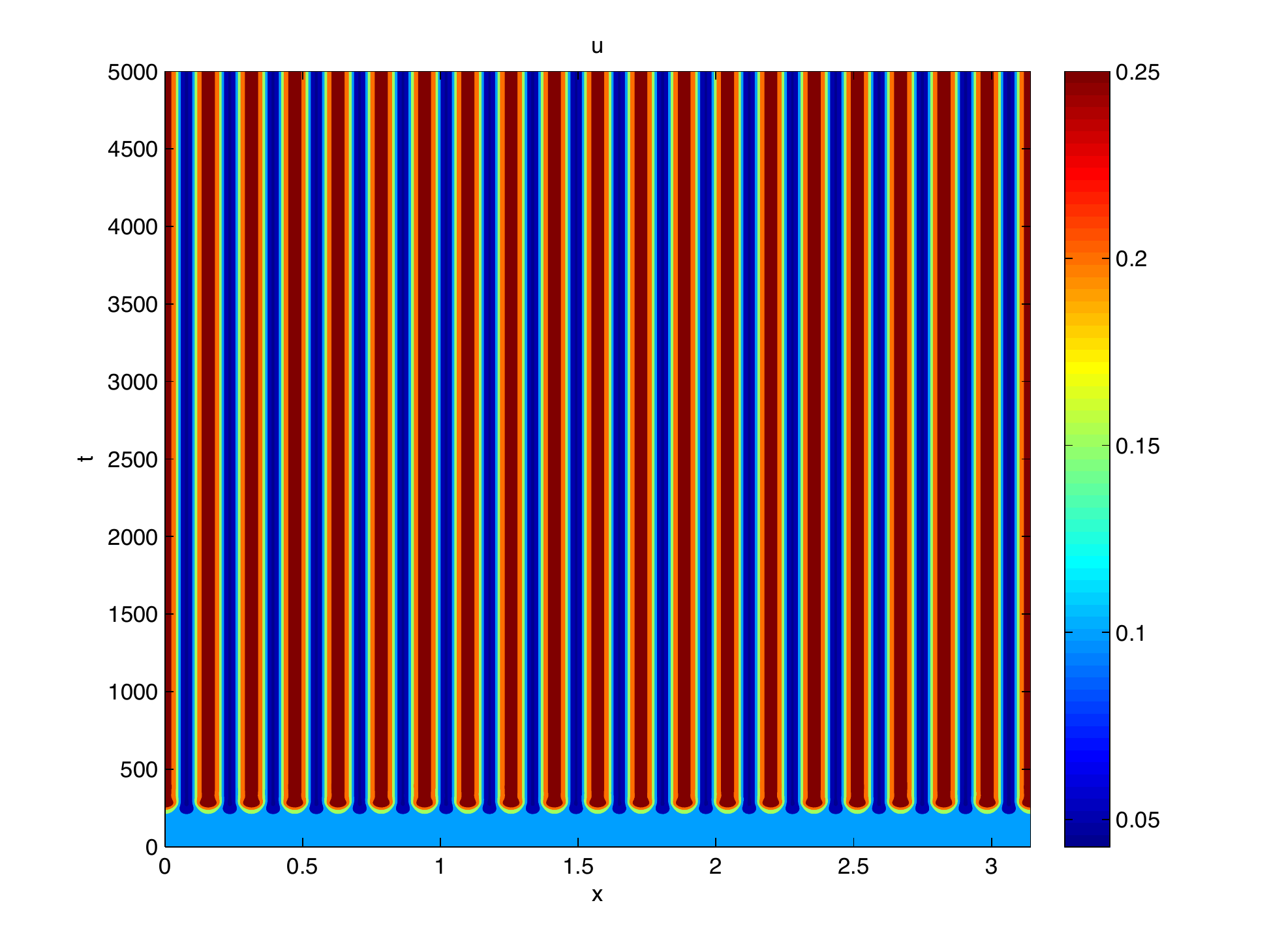}  
\includegraphics[scale=0.17]{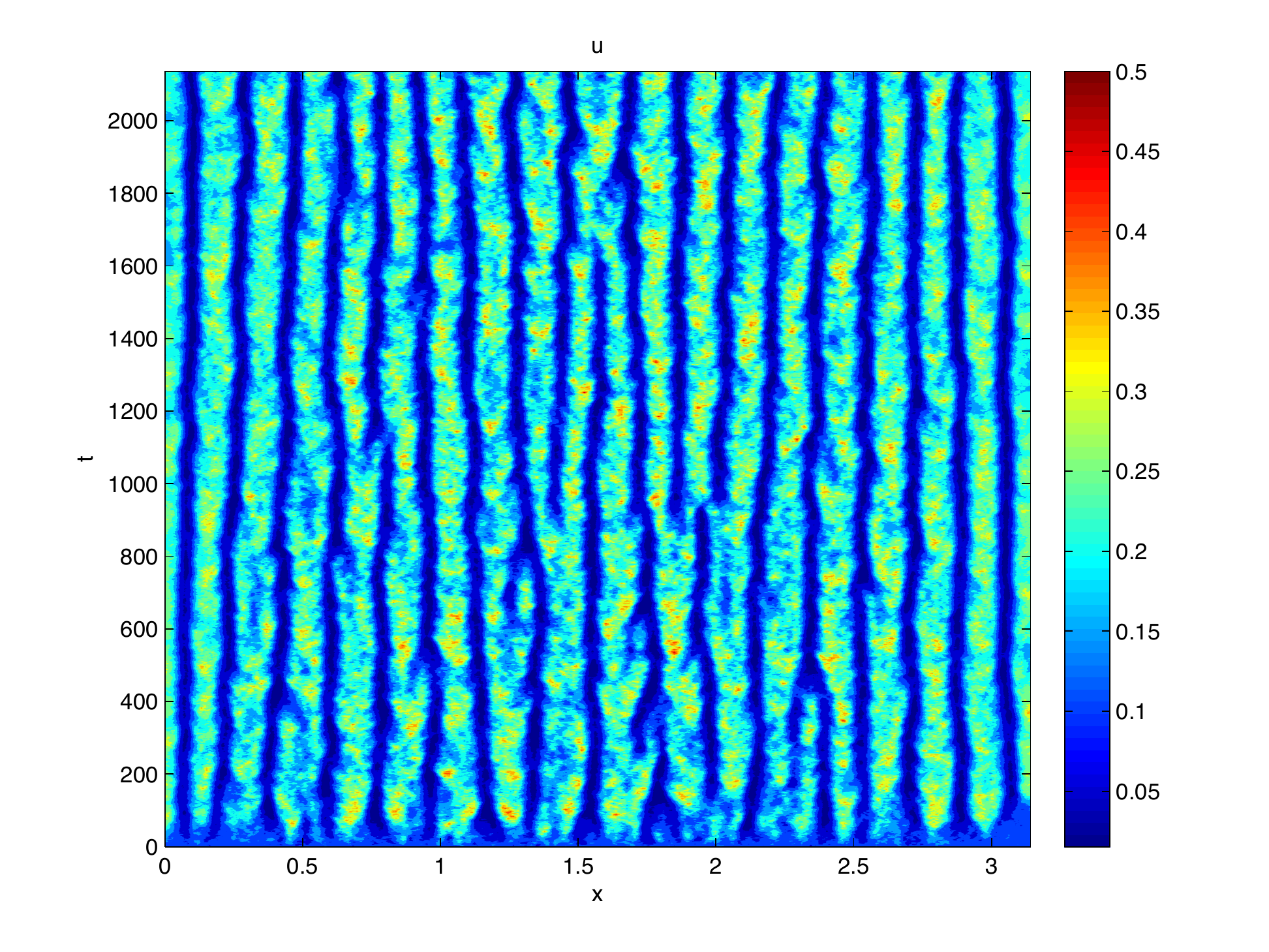}  %
\includegraphics[scale=0.17]{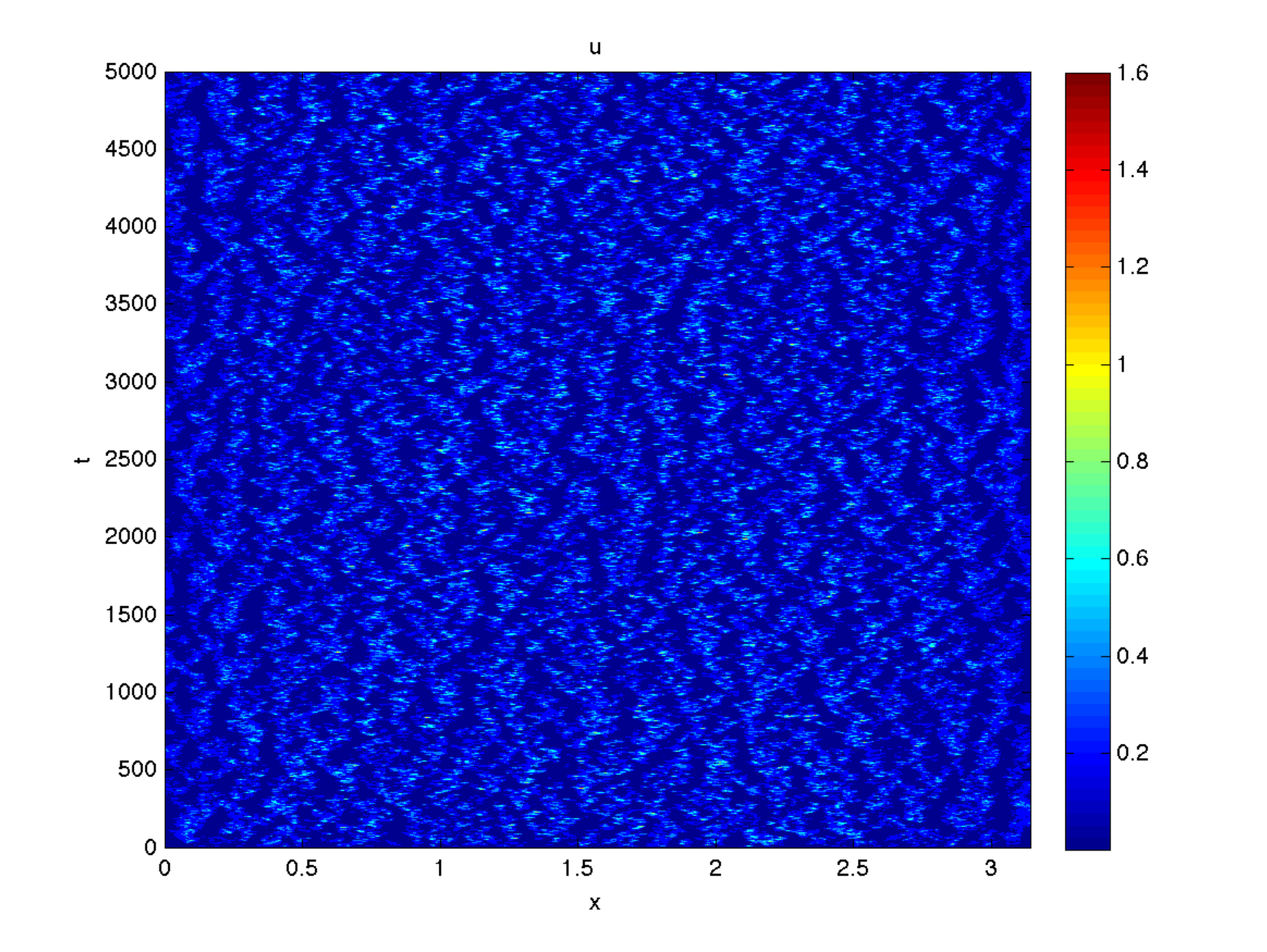}  
\end{center}
\caption{}
\label{fig:TuringIntensity}
\end{figure}
We next investigate whether noise can induce pattern formation. This
achieved by using a non-Turing parameter set. Fig.%
\eqref{fig:NoTuringIntensity} shows the dynamic behavior of model %
\eqref{equ:MB model} with non-Turing parameter set. The first plot of Fig.%
\eqref{fig:NoTuringIntensity} shows the plot of the deterministic model
where with diffusion the prey population still goes into a steady
state(hence no Turing Patterns formed), the proceeding plot of Fig.%
\eqref{fig:NoTuringIntensity} shows the evolution from non-Turing patterns
to some form of weak patterns as $\epsilon$ changes from $\epsilon=0.1$ to $\epsilon=0.4$. 
\begin{figure}[!htp]
\begin{center}
\includegraphics[scale=0.17]{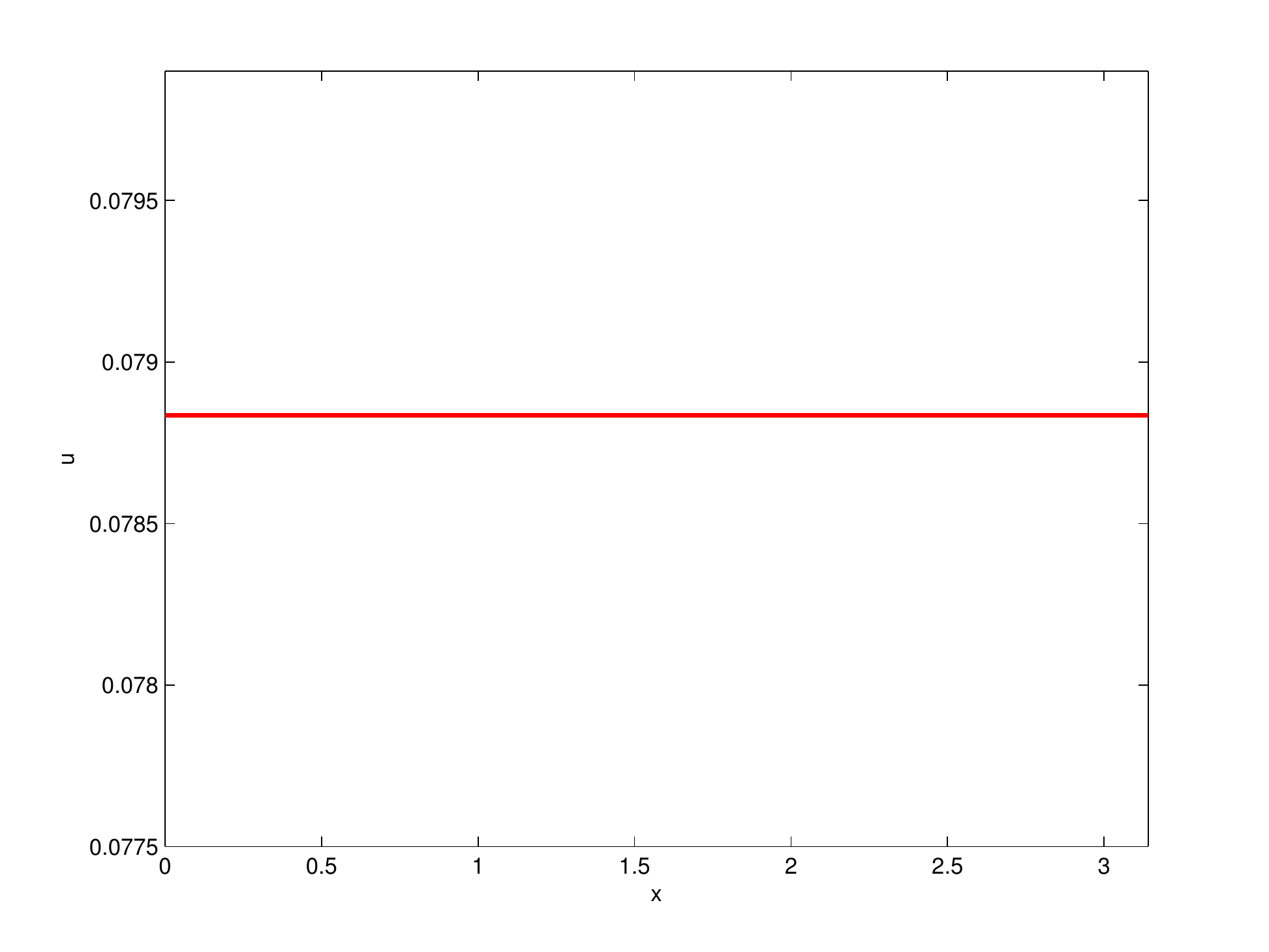}  
\includegraphics[scale=0.17]{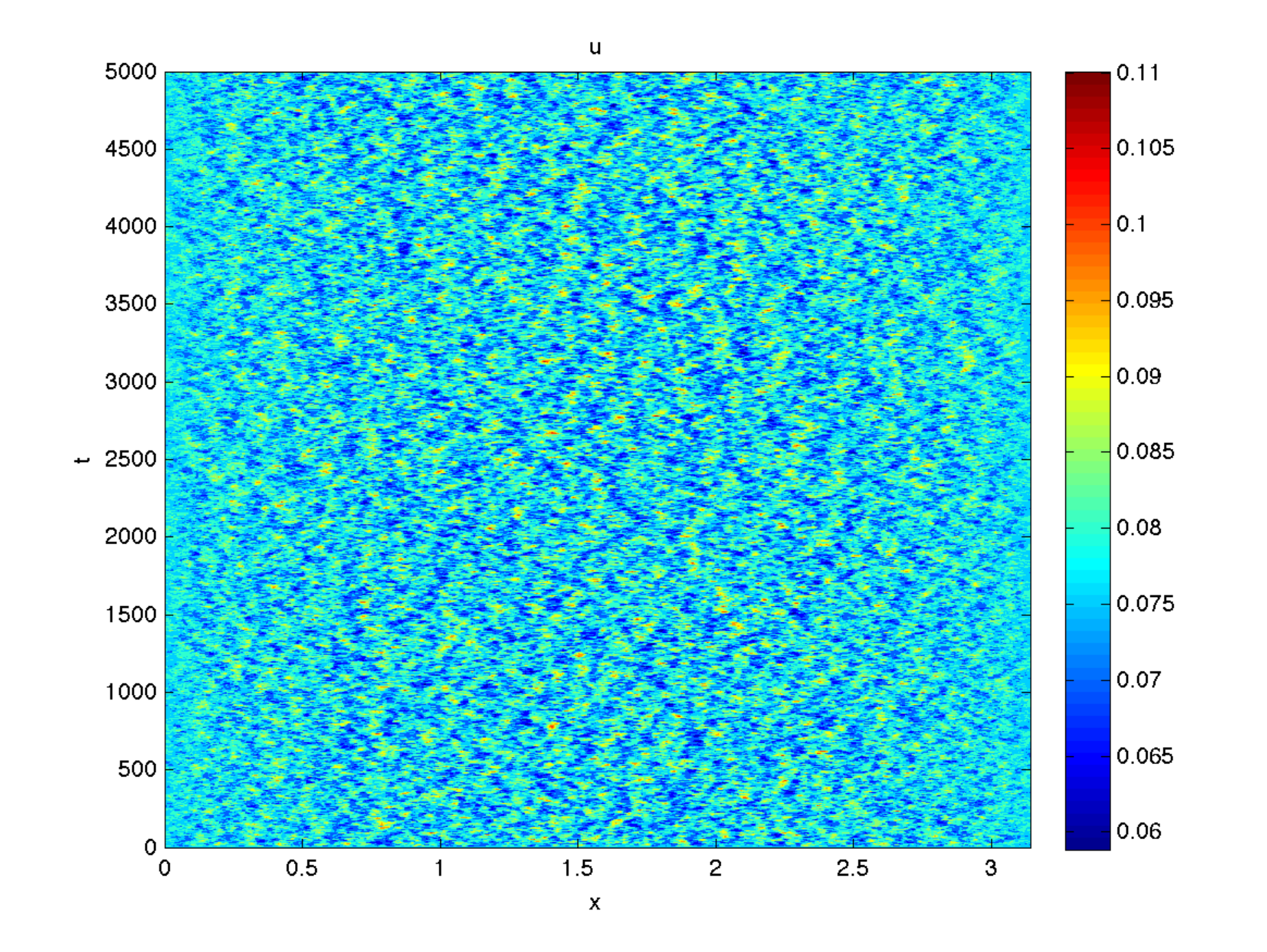}  %
\includegraphics[scale=0.17]{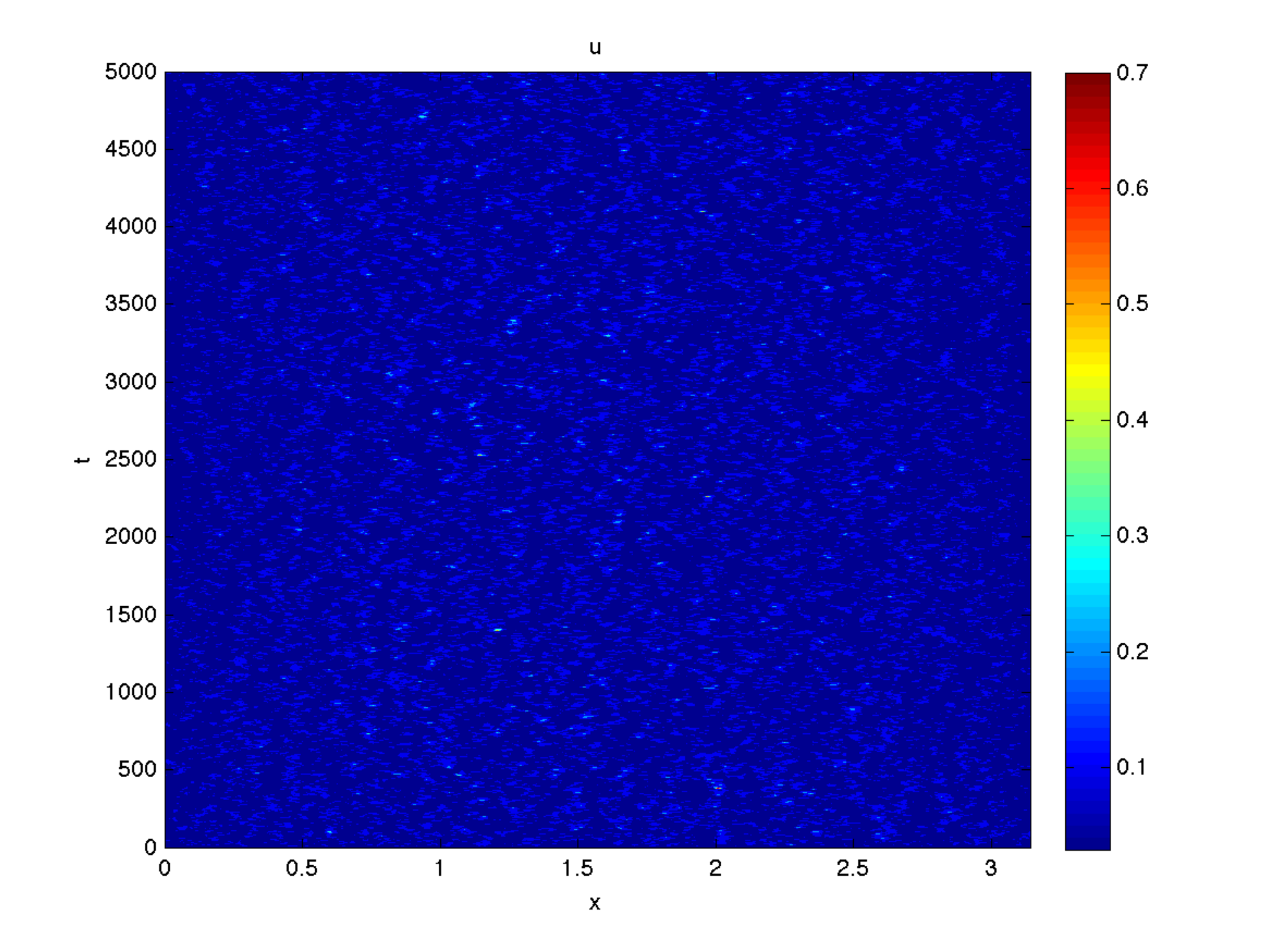}  
\end{center}
\caption{ }
\label{fig:NoTuringIntensity}
\end{figure}

Clearly from above, it is comprehensible that noise can induce instability,
and the effect of noise on model system \eqref{equ:MB model} is drastic.
Thus noise can enhance, as well as cause the disappearance of patterns
formed in the deterministic model. These behaviors confirm how sensitive model
system \eqref{equ:MB model} is to different noise intensities. 
\begin{table}[ht]
\caption{Parameters used in the simulations. In all figures.}
\label{table:Paramset}\centering
\begin{tabular}{ccccccccc}
\hline\hline
Fig. & $\alpha$ & $\beta$ &$c_1$& $c$ & $d$ & $\delta$ & $D_u$ & $D_v$ \\[0.5ex] \hline
\ref{fig:2s},\ref{fig:3s},\ref{fig:4s}& 0.73 & 5.21 & 0.25 & 1.0& 0.11 & 0.013 & 0.0 & 0.0 \\  
\ref{fig:nc} &  0.9&2.2  &0.01 &0.25  &0.1  & 0.1&0.0 &0.0 \\ 
\ref{fig:2},\ref{fig:3}a & 1.0 & 7.35 & 0.01 & 0.17& 0.01 & 0.32 & $10^{-5}$ & $10^{-2}$ \\ 
\ref{fig:3}b & 1.0 & 7.35 & 0.0 & 0.0& 0.00 & 0.32 & $10^{-5}$ & $10^{-2}$ \\ 
\ref{fig:TuringIntensity}  & 1.0 &2.72  &0.01  & 0.17 & 0.01 & 0.32 & $10^{-5}$ & $10^{-2}$\\  
\ref{fig:NoTuringIntensity}& 1.05 &7.34 &0.01  & 0.2 & 0.1 & 0.32 & $10^{-5}$ & $10^{-2}$ \\[1ex] \hline 
\end{tabular}%
\end{table}
\section{Discussion and Conclusion}
Cannibalism, although ubiquitous in both predator and prey communities, has primarily been studied as a phenomenon in the predator species.
In the current manuscript we have considered the effect of both predator and prey cannibalism on a Holling Tanner model, with ratio dependent functional response.
Our results show that in the ODE version of the model, in certain parameter regime, predator cannibalism is able to stabilize the unstable interior equilibrium, (that is the equilibrium is unstable if there is no cannibalism). 
This stability result depends critically on the cannibals feeding rate of prey $\gamma$, and the ratio $\dfrac{ \beta }{\beta_{1}} $. That is if the cannibalistic predator eats less prey so $\gamma$ is small in comparison to $\dfrac{ \beta }{\beta_{1}} $, then only a little bit of feeding of conspecifics is enough to stabilize the dynamics.
This is seen via lemma \ref{prcs}, and is in accordance with known facts about predator cannibalism in other models. However, when we consider prey cannibalism, we see that it is \emph{not} able to stabilize the unstable interior equilibrium. This is rigorously proven for small prey cannibalism rate $c$, via lemma \ref{pcns}. We conjecture this is true even for large cannibalism rate.

This result tells us that for parameters such that the no cannibalism model \eqref{eq:MBmodel_1} has an unstable equilibrium, or essentially cyclical dynamics, no amount of prey cannibalism can stabilize this. From an ecological perspective, this tells us that in a predator-prey community with cyclical fluctuation in the predator and prey populations (so essentially limit cycle dynamics), cannibalism in the prey would just maintain the cyclically, and the oscillations in population cannot be driven to a stable steady state. Interestingly, this is seen to hold even if we assume some stochasticity in the prey cannibalism rate, see figs. \ref{fig:2s}-\ref{fig:4s}.

We also see that prey cannibalism can lead to limit cycle dynamics, but the uniqueness or non uniqueness of limit cycles is an open question. It would be extremely interesting to investigate this further. In particular one may ask, if changing the form of the prey cannibalism term $C(u)$ may lead to multiple limit cycles, or render the limit cycle unique. Another interesting open problem, would be to consider the effect of \emph{both} predator and prey cannibalism together, on models of Holling Tanner type. In particular, one might try to classify regions, where they have conflicting effects, that is say predator cannibalism stabilises while prey cannibalism destabilises, and to then consider the net effect, in terms of the parameters concerned. To the best of our knowledge, this has not been considered in any work in the literature. Also, of much interest lately is the problem of a diseased predator in predator-prey models with cannibalistic predators. It is seen that cannibalism in the predator, can prevent disease in predator from spreading, as the cannibal could wipe out the diseased predator \cite{biswas2014cannibalism}. It would be interesting to consider such a system, with cannibalism now in the prey, and also a diseased prey. Here one can ask if prey cannibalism could stop disease in the prey species, as is known from similar models in predator cannibalism.

In case of the spatially explicit model, we uncover certain interesting results as well. What we observe is that for small predator growth rate, or if prey handling time and prey capture rate are close, Turing patterns are an impossibility in the no cannibalism model. That is spatial segregation cannot occur in \eqref{eq:MBmodel_1}. However, even if one maintains this small predator growth rate, or if prey handling time and prey capture rate are close, if the prey cannibalism rate is chosen in the right regime, see theorem \ref{Thm3123}, then one can \emph{still} have spatial patterns in the prey cannibalism model \eqref{eq:MBmodel_New}. From a spatial ecology perspective this is quite interesting, as this is telling us that if the predator growth rate is curtailed for some reason, in the Holling Tanner model, spatial segregation cannot occur, unless prey cannibalism takes place, in which case spatial segregation can occur. Interestingly, if one compares the destabilizing influence of prey cannibalism in the ODE and PDE models, we find that if the cannibalism rate $c < 2(u_{nc}-u_{c})$, then it will destabilize in the ODE case. However, to destabilize in the PDE case, it cannot be ``to small", and must be such that 
\begin{equation*}
\left( \dfrac{\beta }{\left( 1+\alpha \beta \right) ^{2}}-u_{c}\right) 
\dfrac{\left( u_{c} +d\right) ^{2}}{du_{c}}>{c}%
>\left( \dfrac{\beta }{\left( 1+\alpha \beta \right) ^{2}}-u_{c}- \delta
\beta \right) \dfrac{\left( u_{c}+d\right) ^{2}}{du_{c}} 
\end{equation*}
as seen in theorem \ref{Thm3123}. Alternatively, we show that predator cannibalism cannot cause Turing patterns if the predator feeding rate is curtailed, according to the parametric restriction $\gamma < \frac{\beta}{\beta_{1}}$, as seen in lemma \ref{lemnt1}.

Lastly, ecosystems are always subject to random forcing. However there have been no studies with a random cannibalism rate, even as far as predator cannibalism goes, let alone prey cannibalsim. Our results with a random cannibalism rate show that it can affect both limit cycle dynamics, as well as completely alter the spatial Turing patterns, that form in the deterministic case, see figs. \ref{fig:TuringIntensity}-\ref{fig:NoTuringIntensity}. In case of limit cycle dynamics, noisy cannibalism makes things aperiodic, but cannot remove the population cycles completely. In the case of the spatially explicit model, the noise is seen to both be able to create patterns, as well as to destroy them.

All in all, we hope that the current work leads to many further investigations, into the fascinating effects of prey cannibalism on predator prey dynamics. Such endeavours will bring us closer to understanding the full extent of cannibalism as it occurs in nature, and the role it plays in shaping predator-prey communities.

\section{Acknowledgements}
We would like to acknowledge very helpful conversations with Professor Volker Rudolf, in the department of Biosciences at Rice University, on various ecological concepts of prey cannibalism, and its mathematical modeling.

\bibliographystyle{plain}
\bibliography{Reference}

\end{document}